\documentclass{article}

\usepackage{microtype}
\usepackage{graphicx}
\usepackage{subfigure}
\usepackage{booktabs} % for professional tables

\usepackage{hyperref}

\usepackage[accepted]{icml2023}

\usepackage{amsmath}
\usepackage{amssymb}
\usepackage{mathtools}
\usepackage{amsthm}
\usepackage[]{empheq}
\usepackage{thm-restate}

\usepackage{bbm}
\newcommand{\defeq}{\coloneqq}
\newcommand{\Reals}{\mathbb{R}}
\newcommand{\Naturals}{\mathbb{N}}
\newcommand{\etavec}{\boldsymbol{\eta}}
\newcommand{\xvec}{\boldsymbol{x}}
\newcommand{\dvec}{\boldsymbol{d}}
\newcommand{\Amat}{\boldsymbol{A}}
\newcommand{\Mmat}{\boldsymbol{M}}
\newcommand{\lambdavec}{\boldsymbol{\lambda}}
\newcommand{\argmax}{\operatornamewithlimits{argmax}}

\newcommand{\N}{\mathcal{N}}
\newcommand{\Si}{\mathcal{S}}
\newcommand{\svec}{\boldsymbol{s}}
\newcommand{\xivec}{\boldsymbol{\xi}}
\newcommand{\cvec}{\boldsymbol{c}}
\newcommand{\pivec}{\boldsymbol{\pi}}
\newcommand{\rvec}{\boldsymbol{r}}
\newcommand{\X}{\mathcal{X}}
\newcommand{\A}{\mathcal{A}}
\newcommand{\C}{\mathcal{C}}
\newcommand{\LP}{\text{LP}}

\newcommand{\full}{\textsc{f}}
\newcommand{\bandit}{\textsc{b}}
\newcommand{\E}{\mathcal{E}}
\let\mc\mathcal

\newcommand{\med}{\textsc{r}}
\newcommand{\ag}{\textsc{s}}

\newcommand{\ia}{\textsc{x}}
\newcommand{\ib}{\textsc{y}}
\definecolor{mygreen}{rgb}{0.0, 0.5, 0.0}

\newcommand{\commentsymbol}{\it\color{gray}$\triangleright$~}
\renewcommand\algorithmiccomment[1]{\hfill{\commentsymbol#1}}

\makeatletter
\newcommand{\LComment}[1]{{\commentsymbol{#1}}}
\newcommand{\Statei}{\Statex \hskip\ALG@thistlm}
\makeatother

\usepackage[capitalize,noabbrev]{cleveref}

\usepackage{xparse}

\theoremstyle{plain}
\newtheorem{theorem}{Theorem}[section]

\newtheorem{lemma}{Lemma}[section]
\newtheorem{corollary}[theorem]{Corollary}
\theoremstyle{definition}
\newtheorem{definition}{Definition}[section]
\newtheorem{assumption}{Assumption}
\theoremstyle{remark}

\usepackage[textsize=tiny]{todonotes}

\icmltitlerunning{Online Mechanism Design for Information Acquisition}

\begin{document}

\twocolumn[
\icmltitle{Online Mechanism Design for Information Acquisition}
% It is OKAY to include author information, even for blind
% submissions: the style file will automatically remove it for you
% unless you've provided the [accepted] option to the icml2023
% package.

% List of affiliations: The first argument should be a (short)
% identifier you will use later to specify author affiliations
% Academic affiliations should list Department, University, City, Region, Country
% Industry affiliations should list Company, City, Region, Country

% You can specify symbols, otherwise they are numbered in order.
% Ideally, you should not use this facility. Affiliations will be numbered
% in order of appearance and this is the preferred way.
\icmlsetsymbol{equal}{*}

\begin{icmlauthorlist}
\icmlauthor{Federico Cacciamani}{poli}
\icmlauthor{Matteo Castiglioni}{poli}
\icmlauthor{Nicola Gatti}{poli}
\end{icmlauthorlist}

\icmlaffiliation{poli}{Politecnico di Milano, Milan, Italy}

\icmlcorrespondingauthor{Federico Cacciamani}{federico.cacciamani@polimi.it}

% You may provide any keywords that you
% find helpful for describing your paper; these are used to populate
% the "keywords" metadata in the PDF but will not be shown in the document
\icmlkeywords{Machine Learning, ICML}

\vskip 0.3in
]

% this must go after the closing bracket ] following \twocolumn[ ...

% This command actually creates the footnote in the first column
% listing the affiliations and the copyright notice.
% The command takes one argument, which is text to display at the start of the footnote.
% The \icmlEqualContribution command is standard text for equal contribution.
% Remove it (just {}) if you do not need this facility.

\printAffiliationsAndNotice{}  % leave blank if no need to mention equal contribution
%\printAffiliationsAndNotice{\icmlEqualContribution} % otherwise use the standard text.

\begin{abstract}
	We study the problem of designing mechanisms for \emph{information acquisition} scenarios. This setting models strategic interactions between an uninformed \emph{receiver} and a set of informed \emph{senders}. In our model the senders receive information about the underlying state of nature and communicate their observation (either truthfully or not) to the receiver, which, based on this information, selects an action. Our goal is to design mechanisms maximizing the receiver's utility while incentivizing the senders to report truthfully their information. 
	First, we provide an algorithm that efficiently computes an optimal \emph{incentive compatible} (IC) mechanism. 
	Then, we focus on the \emph{online} problem in which the receiver sequentially interacts in an unknown game, with the objective of minimizing the \emph{cumulative regret} w.r.t. the optimal IC mechanism, and the \emph{cumulative violation} of the incentive compatibility constraints. We investigate two different online scenarios, \emph{i.e.,} the \emph{full} and \emph{bandit feedback} settings. For the full feedback problem, we propose an algorithm that guarantees $\tilde{\mc O}(\sqrt T)$ regret and violation, while for the bandit feedback setting we present an algorithm that attains $\tilde{\mc O}(T^{\alpha})$ regret and $\tilde{\mc O}(T^{1-\alpha/2})$ violation for any $\alpha\in[1/2, 1]$. Finally, we complement our results providing a tight lower bound.
\end{abstract}
\section{Introduction}\label{sec:intro}
Information plays a crucial role in every decision-making process.
The study of how to exploit information to shape the behavior of other self-interested agents has received a terrific attention in recent years under the Bayesian persuasion framework, with application to online advertising~\citep{bro2012send}, voting~\citep{alonso2016persuading,castiglioni2019persuading,semipublic}, traffic routing~\citep{bhaskar2016hardness,castiglioni2020signaling}, and  security~\citep{rabinovich2015information,xu2016signaling}.
Bayesian persuasion~\citep{kamenica2011bayesian} studies how an informed agent (sender) can exploit an information advantage to influence the behavior of an uninformed agent (receiver).
In particular, the sender can commit to an information-disclosure policy (signaling scheme) in order to persuade the receiver to play an action that is desirable for the sender.
In this work we study the inverse problem faced by a receiver that wants to incentivize multiple senders to share their private information about a common state of nature. Surprisingly, while many real-world scenarios are affected by an increasing decentralization of information, the problem of strategically acquiring such information has received little attention from the scientific community. In particular, we pose the question:
\begin{quote}
\emph{Can a receiver with commitment power induce  senders to reveal their information truthfully?}
\end{quote}
In our model, there are multiple senders that receive noisy information (signals) about a common state of nature. Then, each sender strategically reports a signal to the receiver, which chooses an action that depends on the signals reported by the senders.
This action is drawn from a randomized mechanism to which the receiver commits and depends on the signals reported by the senders.
Finally, the senders and the receiver receive an utility that depends on the state of nature and the action played by the receiver.
In this work, we study the problem of designing an (approximately) incentive compatible mechanism that maximizes the sender's utility.
Moreover, we consider an \emph{online} scenario in which the receiver has partial or no information about the parameters of the game. Our goal is to design algorithms that  repeatedly interact with the senders guaranteeing sublinear regret and violations of the incentive compatibility constraints.

\subsection{Original Contribution}
We study the design of Incentive Compatible (IC) mechanisms for the information acquisition problem, focusing on the case in which the senders are \emph{symmetric}, \emph{i.e.}, they receive signals sampled by a common signaling scheme.
First, we show that, under a mild assumption, the optimal mechanism can be computed efficiently. 
In doing so, we provide a class of succinctly representable \emph{symmetric} mechanisms.
This is a non-trivial task since  a general mechanism should specify a distribution over actions for each possible tuple of signals, which are exponentially-many.
Then, we study the online problem in which the receiver does not know the parameters of the game and need to learn them during a repeated interaction with the senders.
The uncertainty on the game parameters comes in two flavors: \emph{full} and \emph{bandit} feedback. 
In the full feedback setting we assume that the receiver has partial knowledge of the game and, at the end of each iteration, can observe the state of nature that was sampled. In bandit feedback scenario, instead  we drop the assumption on the receiver's knowledge of the game and study the problem in which the feedback received by the receiver is restricted to the utility values.
We first present a negative result showing the impossibility of developing an algorithm that guarantees no violation of the IC constraints and sublinear regret on the receiver's utility with respect to the optimal IC mechanism. Thus, we relax our requirements to sublinear regret and sublinear IC violation.
%
%Our goal is to design online algorithms that guarantee sublinear regret on the receiver's utility with respect to the optimal IC mechanism and sublinear violation of the IC constraints.
%
In the full feedback setting, we provide an efficient online algorithm that attains $\tilde{\mc O}(\sqrt{T})$ regret and constraint violation.
We show that achieving the same guarantees is not possible under bandit feedback for which we provide a lower bound frontier on the trade-off between regret and constraint violation. 
This lower bound suggests that optimal algorithm can work in two phases: exploration and exploitation phase.
Indeed, we show that our algorithm achieves $\tilde {\mathcal{O}}(T^{\alpha})$ regret and $\tilde {\mathcal{O}}(T^{1- \alpha/2})$ violation of the IC constraints for any $\alpha\in [1/2,1]$, matching the lower bound frontier.

\subsection{Related Works}
The study of information acquisition has been mostly confined to economics, with particular focus on auctions (see, \emph{e.g.,} \citep{bergemann2002information, bikhchandani2017mechanism, persico2000information}). \citet{stegeman1996participation} studies the problem of acquiring information from a set of buyers in an auction with costly communication, while \citet{shi2012optimal} investigates a setting in which the buyers pay a cost to acquire information. Recently, some works addressed the computational problem of information acquisition. More in particular, a line of work \citet{chen2021optimal,oesterheld2020minimum,papireddygari2022contracts,li2022optimization,neyman2021binary} studies the problem of designing optimal scoring rules to incentivize an agent to acquire and report costly information. Crucially, these works differ from our work in many aspects. For instance, we study how to incentivize agents only by inducing favorable outcomes of the game (\emph{i.e.,} actions), while previous works rely on payments. Moreover, to the best of our knowledge, this paper is the first work addressing the computational problem of information acquisition with multiple agents. 

Our work is also related to the study of Bayesian persuasion in online settings.
Some works study the learning problems in which the receivers’ payoffs are unknown ( see, e.g., \cite{castiglioni2021Online, castiglioni2020online, castiglioni2023regret}).
Other works study Bayesian persuasion in MDPs.
\citet{gan2021Bayesian} show how to efficiently find a sender-optimal policy when the receiver is myopic.  
\citet{Wu2022Sequential} extend the work considering an unknown environment. 
The works closest to ours consider online problems in which the agents do not know the prior over the states of nature.
\citet{DBLP:conf/sigecom/ZuIX21} studies a persuasion problem in which the sender and the receiver do not know the prior and interact online.
\citet{bernasconisequential} extend the analysis to sequential games.

\section{Preliminaries}\label{sec:prelim}
In this section we present the model of interaction between the different agents involved and introduce the main solution concepts on which we rely. 

\paragraph{Game model.} We study the case in which an agent $\med$ (called \emph{receiver}) interacts with a set of agents\footnote{In this work, given $n\in\mathbb{N}_{>0}$, we denote as $[n] = \left\{1, ...., n\right\}$ the set of the first $n$ natural numbers.} ${\N = [n]}$ (called \emph{senders}). At the beginning of the game, a \emph{state of nature} $\theta\in\Theta$ is sampled from a \emph{prior}\footnote{In this work, for any finite set $\mc Z$, we denote as $\Delta(\mc Z)$ the set of probability distributions defined over the elements of $\mc Z$.} $p\in\Delta(\Theta)$, where $\Theta$ is the set of possible states of nature. Neither the receiver, nor the senders, observe the state of nature. Instead, each sender ${i\in\N}$ observes a signal $s_i$ from a specific set of signals $S$. The signal $s_i$ is drawn independently for each $i\in\N$ from a probability distribution dependent on the state of nature, which is specified by the \emph{signaling scheme} $\psi_\ag: \Theta\to\Delta(S)$ that is shared among all the senders. We write $\psi_\ag(s_i\vert\theta)$ to denote the probability with which, when the state of nature is $\theta$, signal $s_i\in S$ is sampled. After observing the signal $s_i$, each sender $i\in\mc N$ communicates some signal $s_i^\prime$ (note that it might be the case that $s_i^\prime\neq s_i$) to $\med$, which, based on the \emph{signal profile} $\svec^\prime = \left(s_1^\prime,....,s_n^\prime\right)$ received, selects an action from a finite set $\mc A$. A mechanism $\xvec$ for the receiver defines a mapping $\xvec: \mc S\to\Delta(\mc A)$ from signal profiles $\svec\in\Si\defeq \times_{i\in\mc N} S$, to probability distributions defined over the set of actions. The set of possible mechanisms can be defined as the polytope $\X$ composed of vectors\footnote{We denote vectors with bold symbols. For any finite set $D$, $\boldsymbol{v}\in\mathbb{R}^{|D|}$ denotes a $|D|$-dimensional vector indexed over $D$, where for each $d\in D$, $v[d]$ is the value of $\boldsymbol{v}$'s component corresponding to $d\in D$.} indexed over pairs $(\svec, a)\in \Si\times \A$ such that\footnote{We drop the parentheses and denote $x[(s,a)]$ as $x[s,a]$.}
\[
	\X \defeq \left\{\xvec\in\Reals_{\geq 0}^{|\Si||\A|}\mid \sum_{a\in\A} x[\svec, a] = 1 \quad \forall \svec\in\Si \right\}.
\] 
Given a signal profile $\svec\in\Si$, we denote as $\svec_{-i}$ the signal profile obtained from $\svec$ by removing the $i$-th element $s_i$. Similarly, we denote as $\left(s_i^\prime, \svec_{-i}\right)$ the signal profile obtained from $\svec$ by substituting signal $s_i$ with signal $s_i^\prime$.
After the receiver's choice of an action $a\in\A$, she receives a payoff $u_\med(a, \theta)\in [0, 1]$, while each sender $i\in\mc N$ receives an equal payoff $u_i(a, \theta) = u_\ag(a, \theta)\in [0, 1]$.
Throughout this work, we will make the following assumption:
\begin{assumption}
	The number of signals $|S|$ is a constant.\label{ass:const}
\end{assumption}
Assumption~\ref{ass:const} is required to represent succinctly the set of possible mechanisms. Indeed, in general a mechanism must specify an action distribution for each signal profile in $\Si$, which are exponentially many. Such mechanisms has an exponentially-large representation and cannot be represented efficiently.
In the following Section, we show that thanks to the symmetry among the senders, symmetric mechanisms, \emph{i.e.}, which recommend the same action distribution for symmetric signal profiles, are optimal.
However, representing symmetric mechanisms requires space exponential in $|S|$. Our assumption is required to efficiently represent symmetric mechanisms.
It is not clear whether there exists a class of optimal mechanisms that can be represented in polynomial space when $|S|$ is not constant.

\paragraph{Incentive compatibility.} 
In this work, we are interested in characterizing the subset of receiver's mechanisms such that all the senders are incentivized to report truthfully the information they have (\emph{i.e.,} the signal they observed). To this extent, we consider a finite set of \emph{deviation functions} $\Phi$ for the senders, where each $\varphi\in\Phi$ defines a mapping $\varphi: S\to S$ that specifies the signal that the sender is going to report to $\med$, following every possible signal observation. 
We are interested in characterizing the set of receiver's mechanisms -- denoted as \emph{Incentive Compatible} (IC) mechanisms -- that are stable with respect to unilateral deviations of senders in reporting the information they have. More in detail, we require the expected utility of each sender reporting truthfully their information not to be smaller than the expected utility she would get following any deviation function $\varphi\in\Phi$. 
When $\med$ uses mechanism $\xvec\in\X$ and senders report truthfully the signals they observe, the expected utility that each agent $i\in\mc N\cup\left\{\med\right\}$ gets is defined as:
\[
U_i(\xvec) \defeq \sum_{\substack{\svec\in\Si\\ a\in\A}} x[\svec, a] \sum_{\theta\in\Theta} p(\theta) u_i(a, \theta) \psi_\ag(\svec\vert\theta),
\]
where, with an overload of the notation, we denoted ${\psi_{\ag}(\svec\vert\theta) = \prod_{j\in\mc N}\psi_\ag(s_j\vert\theta)}$.
Similarly, the expected utility of sender $i\in\mc N$, when she deviates following deviation function $\varphi\in\Phi$ and all the others behave honestly is the following: 
\[
		U_i^\varphi(\xvec) \defeq \sum_{\substack{\svec\in\Si\\ a\in\A}} x[(\varphi(s_i),\svec_{-i}), a] \sum_{\theta\in\Theta} p(\theta) u_i(a, \theta) \psi_\ag(\svec\vert\theta).
\]
We are now ready to formally define the set of IC mechanisms.
\begin{definition}[$\varepsilon$-IC mechanisms]
	For $\varepsilon\geq 0$, the set of $\varepsilon$-IC mechanisms is defined as the polytope
	\[
		\X_\varepsilon \defeq \left\{\xvec\in\X\mid U_i^\varphi(\xvec) - U_i(\xvec) \leq \varepsilon\,\,\forall i\in\mc N,\,\forall \varphi\in\Phi\right\}.
	\]
	The set of IC mechanisms is the set $\X_0$.
\end{definition}

\paragraph{\emph{Ex-ante} and \emph{ex-interim} deviations.} 
The set of deviations $\Phi$ can be suitably specified in order to capture different concepts of incentive compatibility. In this work, we are interested in two particular types of incentive compatibility, namely \emph{ex-ante} incentive compatibility and \emph{ex-interim} incentive compatibility. In the case of ex-ante incentive compatibility, each sender can choose to deviate \emph{only before} observing the signal, thus she is not able to discriminate between the different signals actually sampled when choosing which signal to report. Hence, in this case, the set of possible deviations coincides with the set ${\Phi^{ante}\defeq\left\{\varphi\mid \varphi: S\to S, \varphi(s) = \varphi(s^\prime)\,\,\forall s, s^\prime\in S\right\}}$. Instead, when considering the notion of ex-interim incentive compatibility, the senders first observe the signal sampled, and then decide which signal to report to $\med$. Thus, the set of possible deviations for the ex-interim case can be formulated as ${\Phi^{inter}\defeq\left\{\varphi\mid\varphi\in S\to S\right\}}$. Crucially, while the number of ex-ante deviation functions is $|\Phi^{ante}| = |S|$, the number of possible ex-interim deviation functions is significantly larger, namely $|\Phi^{inter}| = |S|^{|S|}$. However, when the objective is to guarantee ex-interim incentive compatibility, similarly to \cite{greenwald2011no}, it is possible to show the following\footnote{All the proofs are provided in the appendix.}:
\begin{restatable}{proposition}{propgreen}\label{prop:dev}
	For all $s \in S$ let $\overline\Phi(s)$ be the set of deviation functions such that
	\[
	\overline\Phi(s)\defeq\left\{
	\varphi\mid\varphi\in S\to S, \varphi(s^\prime) = s^\prime,\quad\forall s^\prime\in S\setminus \{s\}
	\right\}. 
	\]
	Then, for all $\xvec\in\X$ and $\varepsilon\geq 0$ , if $\xvec$ is $\varepsilon$-IC with respect to the set of deviation functions $\overline\Phi\defeq \cup_{s\in S}\overline\Phi(s)$, then it is also $|S|\varepsilon$-IC with respect to deviation functions $\Phi^{inter}$.
\end{restatable}
Since the number of deviation functions in  ${\overline\Phi = \cup_{s\in S}\overline\Phi(s)}$ is bounded by $|\overline\Phi| = |S|^2 - |S| + 1$, as a byproduct of Proposition \ref{prop:dev}, we have that it is possible to efficiently represent the needed deviation functions also when the incentive compatibility concept of interest is the ex-interim one. 

Throughout the rest of this paper, we will analyze the problem of information acquisition for a general incentive compatibility concept and thus for a general set of deviation functions $\Phi$. The results that we obtain can be easily adapted to capture both ex-ante and ex-interim incentive compatibility by instantiating the set $\Phi$ to $\Phi^{ante}$ and $\overline\Phi$, respectively.

\section{Offline Information Acquisition}\label{sec:compact}

We are interested in cases in which the objective of the receiver is twofold: (i) on the one hand she is interested in incentivizing the senders to report truthfully their private information, (ii) while on the other hand she is interested in maximizing her expected utility when all the senders behave truthfully. Formally, such a twofold objective can be modeled in terms of a linear optimization problem where the feasibility set is restricted to IC mechanisms (i), and the objective is the maximization of the receiver's expected utility $U_\med$ (ii).  Formally,
\begin{definition}[Optimal $\varepsilon$-IC mechanism]
	An optimal $\varepsilon$-IC mechanism $\xvec^\star$ is defined as:
	\[
	\xvec^\star\in \argmax_{\xvec\in\X_\varepsilon} U_\med\left(\xvec\right).
	\] 
	An optimal IC mechanism is an optimal $0$-IC mechanism.
\end{definition}

In this Section, we investigate the problem of computing an optimal IC mechanism when all the game parameters (\emph{i.e.,} the utility functions, the prior and the signaling scheme) are known. The first problem that we investigate concerns the representation of the polytope $\X$. The trivial way of representing such space would require for each vector $\xvec\in\X$ a number of entries exponential in the number of agents $n$, since the number of possible signal profiles is $|\Si| = |S|^n$. We tackle this problem by showing that, when the number of signals $|S|$ is a constant, in order to compute an optimal IC mechanism, it is possible to restrict our attention to a subset of $\X$ that can be represented succinctly.

\subsection{Compact Formulation of the Set of Mechanisms}

The rationale behind our construction exploits the fact that, since the senders are symmetric (\emph{i.e.,} they share the same utility functions, signal spaces and signaling schemes), it is possible to safely aggregate different signal profiles, thus reducing the number of variables needed to describe the mechanisms' polytope.
In order to characterize this notion of equivalence between signal profiles, we introduce the concept of \emph{symmetric signal profiles}\footnote{In this work, $\mathbbm{1}[\cdot]$ denotes  the indicator function.}.
\begin{definition}[Symmetric signal profiles]
	Two signal profiles $\svec, \svec^\prime\in\Si$ are \emph{symmetric} (in symbols $\svec \bowtie \svec^\prime$) if 
	\[
		\sum_{i\in\mc N} \mathbbm{1}\left[s_i = s\right] = \sum_{i\in\mc N} \mathbbm{1}\left[s^\prime_i = s \right]\quad\forall s\in S.
	\]
\end{definition}
In other words, two signal profiles are symmetric if each signal $s \in S$ occurs the same number of times in both signal profiles.
The $\bowtie$ relationship defines a partition of elements in $\Si$. 
In particular, we can define a partition ${\C = \left\{c_1, ..., c_K\right\}}$ of $\Si$ such that for each $c_i\in\C$, $\svec,\svec^\prime\in c_i$ if and only if $\svec\bowtie\svec^\prime$. 
Given a signal profile $\svec\in\Si$, we denote with $c(\svec)\in\C$ the element in $\C$ such that $\svec\in c(\svec)$.
Moreover, the partition $\C$ can be used to define the set of \emph{symmetric mechanisms}, which contains all the mechanisms that associate the same probability distribution over actions to signal profiles belonging to the same class\footnote{Each element $c\in\C$ will be also called \emph{class}.}. Formally, 
\begin{definition}[Symmetric mechanisms]	\label{def:symm_mech}
	The set of symmetric mechanisms is defined as the set $\X^\circ\subseteq\X$ such that:
	\[
	\X^\circ \defeq\left\{\xvec\in\X\mid x[\svec, a] = x[\svec^\prime, a]\quad\forall \svec\bowtie\svec^\prime,\,\forall a\in\A\right\}.
	\]
\end{definition}  
The following Theorem provides a key result concerning the optimality of symmetric mechanisms.
\begin{restatable}{theorem}{thsymmetry}	\label{th:thsymmetry}
	For any $\varepsilon\geq 0$, there exists a symmetric mechanism $\xvec^\circ\in\X^\circ$ such that 
	\[
		\xvec^\circ\in \argmax_{\xvec\in\X_\varepsilon} U_\med(\xvec).
	\]
\end{restatable}

Theorem \ref{th:thsymmetry} states that, an optimal IC mechanism can be found by restricting our attention to symmetric mechanisms only. To this extent, we provide a compact formulation of the polytope of symmetric mechanisms as follows: 
\[
\Xi \defeq \left\{\xivec\in\Reals^{|\C||\A|}_{\geq 0}\mid \sum_{a\in\A}\xi[c, a] = 1 \quad\forall c\in\C \right\}.
\]
We also define the set of \emph{deterministic symmetric mechanisms} which is defined as the set $\Pi$ of mechanisms that associate deterministically an action to each class $c\in\C$. Formally, $\Pi\defeq\Xi\cap\left\{0, 1\right\}^{|\C||\A|}$.
Crucially, differently than $\X$, as we show in the following Lemma, the polytope $\Xi$ can be represented efficiently.

\begin{restatable}{lemma}{lemmapoly}
	The polytope $\Xi$ admits a description of size polynomial in the dimension of the game.
	\label{le:lemmapoly}
\end{restatable}
\subsection{Efficient Computation of an Optimal IC Mechanism}
To conclude this Section, we provide a compact formulation of the problem of computing an optimal IC mechanism, leveraging the definition of the polytope $\Xi$ of symmetric mechanisms. 

First, let us introduce the vectors $\rvec_i\in\Reals^{|\C||\A|}_{\geq 0}$, defined as follows:
\[
	r_i[c, a] \defeq \sum_{\svec\in c} \sum_{\theta\in\Theta}p(\theta) u_i(a, \theta)\psi(\svec\vert\theta)\quad\forall c\in\C,\,\forall a \in\A.
\]
Notice that, since all the senders are symmetric (in particular since $\forall i, j\in\mc N$, $u_i=u_j=u_\ag$), for each $i,j\in\mc N$ we have that $\rvec_i = \rvec_j$. As a consequence of this, in order to lighten the notation, we will use a single utility parameter $\rvec_\ag = \rvec_i$, $\forall i\in\mc N$, which is common to all the senders.
Then, trivially, the expected utilities of each sender and of the receiver, when the mechanism used by $\med$ is $\xivec$ and all senders report truthfully the signal they observe, can be formulated, respectively, as the following linear functions in $\xivec$:
\begin{align*}
U_i(\xivec) &= \xivec^\top\rvec_\ag \quad \forall i\in\mc N\quad \textnormal{and} \quad
	U_\med(\xivec) = \xivec^\top \rvec_\med.
\end{align*}

In a similar way, as the following Lemma shows, it is possible to express the expected utility $U_i^\varphi(\xivec)$ of agent $i\in\mc N$, when she deviates according to deviation function $\varphi\in\Phi$ and all the other agents report truthfully the information they have, as a bilinear function of $\xivec$ and $\rvec_\ag$. Formally\footnote{We denote matrices with bold symbols. Given two finite sets $D$ and $E$, $\Mmat\in\Reals^{|D|\times |E|}$ denotes a $|D|$ by $|E|$-dimensional matrix indexed over $D$ and $E$, where $M[d,e]$ is the value of $M$'s component corresponding to the pair $d,e\in D\times E$.}: 
\begin{restatable}{lemma}{lemmadev}
	For each $\varphi\in\Phi$ and $i\in\mc N$, let ${\Amat_i^\varphi\in\Reals^{|\C||\A|\times |\C||\A|}}$ be the matrix such that ${\forall c, c^\prime\in\C}$, ${\forall a, a^\prime\in\A}$, 
	\[
	A_i^\varphi[(c, a), (c^\prime, a^\prime)] = \begin{cases}
		0 &\text{if } a\neq a^\prime \\
		\frac{\sum\limits_{s\in c^\prime} \mathbbm{1}[(\varphi(s_i), \svec_{-i}) \in c]}{|c^\prime|} & \text{otherwise}.
	\end{cases}
	\]
	Then, it holds $U_i^\varphi(\xivec) = \xivec^\top\Amat_i^\varphi\rvec_\ag
	$ for all $i\in\mc N$.
 \end{restatable}
Intuitively, for any $c, c^\prime\in\C$ and any action $a\in\A$, the parameter $A_i^\varphi[(c,a), (c^\prime, a)]$ encodes the probability with which, when sender $i$ deviates according to $\varphi$, $\med$ receives a signal profile belonging to $c$, given that the true signal profile sampled was in $c^\prime$. Again, it is possible to exploit the symmetry between the senders to show the following:
\begin{restatable}{lemma}{lemmasymmdev}
	For any $i, j\in\mc N$ and any $\varphi\in\Phi$ it holds that
	$\Amat_i^\varphi = \Amat_j^\varphi.$
	
\end{restatable}   
As we did above, to ease the notation, given a deviation function $\varphi\in\Phi$, we define a single deviation matrix ${\Amat_\ag^\varphi = \Amat_i^\varphi}$ for all $i\in\mc N$.  
Let us introduce the linear optimization problem $\LP(\varepsilon, \hat\rvec_\med, \hat\rvec_\ag)$, which is defined as follows:
\[
\LP(\varepsilon, \hat\rvec_\med,\hat\rvec_\ag) \defeq
\begin{cases}
	\max\limits_{\xivec\in\Xi} & \hspace{-0.2cm} \xivec^\top \hat\rvec_\med  \\
	&\hspace{-0.2cm} \xivec^\top\left(\Amat_\ag^\varphi \hat\rvec_\ag - \hat\rvec_\ag\right) \leq \varepsilon\,\,\,\forall\varphi\in\Phi.
\end{cases}
\]
The objective function of the LP consists in the maximization of the expected utility of the receiver computed with respect to the utility vector $\hat\rvec_\med$, while the constraints restrict the feasibility set to $\varepsilon$-IC symmetric mechanisms. Crucially, given the structure of the game, the parameters needed to define $\LP(\varepsilon, \hat\rvec_\med,\hat\rvec_\ag)$ can be determined efficiently, thus allowing us to efficiently compute an optimal $\varepsilon$-IC mechanism. 
\begin{restatable}{proposition}{propeff}
	For any deviation function $\varphi\in\Phi$, the coefficients of the matrix $\Amat_\ag^\varphi$ and those of the utility vectors $\rvec_\med$ and $\rvec_\ag$ can be computed in time polynomial in the game size.
\end{restatable}
\begin{corollary}
	For any $\varepsilon\geq 0$, an optimal $\varepsilon$-IC mechanism can be found in time polynomial in the game size by solving $\LP(\varepsilon, \rvec_\med, \rvec_\ag)$. 
\end{corollary}
\section{Online Information Acquisition}\label{sec:online}
In this work, we study the information acquisition problem in a typical \emph{online learning setting}, capturing the scenario in which a receiver $\med$ sequentially interacts with a set of senders in a game which structure might be either partially or completely unknown to $\med$. 

The sequential interaction unfolds as follows.
At each round $t\in [T]$, a state of nature $\theta^{t}$ is sampled according to prior $p$. Then, each sender $i\in\mc N$ observes a signal $s_i^t\sim \psi_\ag(\cdot\vert\theta^t)$ and reports truthfully her observation to $\med$. This latter selects a mechanism $\xivec^t\in\Xi$ and, according to such mechanism, she samples an action $a^t\sim \xivec^t[c(\svec^t), \cdot]$. Finally, each sender receives a payoff $u_\ag^t = u_\ag(a^t, \theta^t)$, while the receiver receives a payoff $u_\med^t = u_\med(a^t, \theta^t)$. 

As discussed in Section \ref{sec:prelim}, the goal of the receiver is the maximization of her own expected utility, while incentivizing the senders to report truthfully the signal they observe. This second objective corresponds to requiring each mechanism $\xivec^t$, $t\in [T]$,  to be IC, while the performances over $T$ rounds for what concerns the first objective are measured by means of the cumulative regret $R^T$, which is defined as 
\[
R^T \defeq \sum_{t\in [T]} \left[\left(\xivec^\star\right)^\top \rvec_\med - \left(\xivec^t\right)^\top \rvec_\med\right],
\]
where $\xivec^\star$ is an optimal IC mechanism. The cumulative regret measures the loss in expected utility suffered by $\med$ for having adopted in the first $T$ rounds mechanisms $\xivec^t$ instead of the optimal IC mechanism $\xivec^\star$. As customary, we require the cumulative regret to grow sublinearly in $T$, \emph{i.e.,} $R^T = o(T)$.

While asking for each mechanism $\xivec^t$ to be IC seems a reasonable requirement, as we show in the following theorem,\color{white}a\color{black} it is not possible to guarantee with high probability that each $\xivec^t$ is IC while attaining sublinear regret.
\begin{restatable}{theorem}{thimp}\label{th:thimp}
	There exists a constant $\delta>0$ such that no algorithm can guarantee to output a sequence of mechanisms $\xivec^1,...,\xivec^T$ such that, with probability at least $1-\delta$ all mechanisms are IC and $R^T=o(T)$. 
\end{restatable}
%
%\fe{Introdurre discussione rebuttal.}
%\color{blue}
%\Cref{th:thimp} has an important implication. Indeed, if the receiver selects a mechanism which is not IC, it becomes fundamentally hard to characterize the behavior of the senders. In particular, the satisfaction of the incentive compatibility constraints guarantees that for the senders it is an equilibrium to behave truthfully. Instead, if the mechanism is not IC, then,  
%\color{black}
Motivated by such an impossibility result, we introduce the concept of cumulative IC violation $V^T$, defined as
\[
V^T \defeq \sum_{t\in[T]}\left[ \max_{\varphi\in\Phi} \left(\xivec^t\right)^\top\left(\Amat_\ag^\varphi\rvec_\ag 
- \rvec_\ag\right)\right].
\]
Intuitively, the cumulative IC violation expresses how much the mechanisms $\xivec^t$ chosen by the receiver at each round $t$ cumulatively violate the IC constraint of $\LP(0, \rvec_\med, \rvec_\ag)$. From another perspective, $V^T$ captures how much utility a sender reporting signals untruthfully would have gained with respect to what she gained by always behaving truthfully. Thus, $V^T$ can also be interpreted as a measure of senders' regret for truthful behavior.
Relaxing the incentive compatibility requirement, our new objective becomes the development of learning algorithms for the receiver that attain cumulative regret and cumulative IC violation that grow sublinearly in $T$, namely $R^T = o(T)$ and $V^T=o(T)$.
%
%\color{blue}
Achieving $V^T = o(T)$ guarantees that, for each sender, truthful behavior converges to the optimal behavior as $T\to\infty$. We remark that, given the negative result of \cref{th:thimp}, it is impossible to guarantee that truthful behavior is an exact best-response to receiver's mechanisms $\xivec^1,...,\xivec^T$.
%\color{black} 

In the following, we investigate two different scenarios, which relate to different degrees of knowledge that the receiver might have on the game structure and to different types of \emph{feedback} that the receiver gets out of the sequential interaction. In particular, in the \emph{full feedback setting} -- which is discussed in Section \ref{sec:full} -- we assume that the receiver knows both utility functions $u_\ag$ and $u_\med$, while she does not know the signaling scheme nor the prior $p$. Furthermore, in such setting, at the end of each round $t\in [T]$, the receiver observes, together with the signal profile $\svec^t$ and the action $a^t$, the state of nature $\theta^t$ that has been sampled. The second scenario that we study is the \emph{bandit feedback setting}, presented in Section \ref{sec:bandit}. In this case, we assume no knowledge of the game structure on the receiver's side (\emph{i.e.,} the receiver does not know the prior $p$, the signaling scheme $\psi_\ag$, and the utility functions $u_\ag$  and $u_\med$) and we restrict the feedback that $\med$ receives at each round $t\in[T]$ to the utility values $u_\ag^t$, $u_\med^t$, the signal profile $\svec^t$ and the action $a^t$.

\section{Online Information Acquisition with Full Feedback}\label{sec:full}

In the full feedback setting, the receiver can leverage her knowledge of the game parameters, as well as the observation of sampled signal profiles and states of nature, to obtain estimates of the products $p(\theta)\psi_\ag(\svec\vert\theta)$ for each $\svec\in\Si$ and $\theta\in\Theta$. Then, we can easily recover estimators for the vectors $\rvec_i$, $i\in\{\med,\ag\}$, with appropriate high-confidence bounds.
As we show in this Section, such an high confidence region can then be exploited in order to define an online learning algorithm capable of guaranteeing $\tilde{\mc O}(\sqrt{T})$ cumulative regret and IC violations with high probability.
In the first part of this section, we show how the information available to the receiver can be used in order to estimate the vectors $\rvec_{\med}, \rvec_{\ag}$, while the second part of this section is devoted to the presentation and analysis of the algorithm for the full feedback setting. 
\subsection{Estimation of $\rvec_{\med}$ and $\rvec_{\ag}$}
Under the full feedback assumption, at each $t\in [T]$, the receiver observes both the signal profile $\svec^t\in \Si$ and the state of nature $\theta^t\in\Theta$. Thus, using the information collected up to round $t$, as well as the knowledge of the utility functions $u_\med$ and $u_\ag$, a possible estimator of the utility vectors is the following for each $i\in\{\med, \ag\}$, for each $c\in\C$ and for each $a\in\A$: 
\begin{equation}\label{eq:estimatorfull}
	\hat r_i^{t+1}[c, a] := \hspace*{-1mm}\frac{1}{t}\sum_{\theta\in\Theta} \hspace{-1mm}u_i(a,\theta)\hspace{-1mm}\sum_{\tau\in [t]} \hspace{-1mm}\mathbbm{1}\left[c(\svec^\tau) = c, \theta^\tau = \theta\right].
\end{equation}
Intuitively, the term $\sum_{\tau\in [t]} \mathbbm{1}\left[c(\svec^\tau) = c, \theta^\tau = \theta\right]/t$ provides an unbiased estimator of the product $p(\theta)\sum_{\svec\in c}\psi_{\ag}(\svec\vert\theta)$, thus allowing us to conclude the following result:

\begin{restatable}{lemma}{lemmaestfull}\label{le:lemmaunbfull}
	For all $t\in [T]$ and $i\in\{\med,\ag\}$, $\hat\rvec_i^t$ is an unbiased estimator of $\rvec_i$. %Formally,
%	\[
%		\mathbb{E}\left[\hat\psi^{t}(s\vert\theta)\right] = \psi_\ag(s\vert\theta)\quad\forall t\in[T],\,\forall s\in S,\,\forall\theta\in\Theta.
%	\]  
\end{restatable}
%\color{black}
%
As a second step, we are interested in deriving an high confidence region around $\rvec_\med$ and $\rvec_{\ag}$. For $\delta\in (0, 1)$, let us introduce the event $\E^\full_\delta$, defined as
\[
\E^\full_\delta\defeq\left\{||\rvec_i - \hat\rvec_i^t||_\infty\leq \varepsilon^t_\delta\,\,\forall t\in[T],\,\forall i\in\{\med,\ag\}\right\},
\]
where 
\begin{equation}\label{eq:eps}
	\varepsilon^t_\delta\defeq \sqrt{\frac{\log\left(4T|\C||\A|/\delta\right)}{2(t-1)}}.
\end{equation}
Then, using standard concentration arguments, it is possible to state the following Lemma.
\begin{restatable}{lemma}{lemmaconcfull}	\label{le:lemmaconcfull}
	For any $\delta\in (0, 1)$, the following holds
	\[
		\mathbb{P}\left(\E^\full_\delta\right) \geq 1-\delta.
	\]
\end{restatable}
\subsection{Algorithm}
\begin{algorithm}
	\caption{Algorithm for the full feedback setting}\label{alg:full}
	\begin{algorithmic}
		\REQUIRE $T\in \mathbb{N}_{>0}$, $\delta\in (0,1)$, $u_\ag$, $u_\med$
		\STATE \textbf{Initialize:}
		\STATE $\quad\hat r_i^1[c, a] \gets 0\quad\forall i\in\{\med, \ag\},\,\forall c\in\C,\,\forall a\in\A$
		\STATE $\quad \nu^1\gets \infty,\,\,\varepsilon^1_\delta\gets \infty$
		\FOR{$t\in [T]$}
		\STATE $\xivec^t\gets$ optimal solution to $\LP(\nu^t,\hat\rvec_\med^t, \hat\rvec_\ag^t)$
		\STATE Use $\xivec^t$ and observe $\theta^t$, $\svec^t$, $a^t$, $u_\med^t$, $u_\ag^t$
		%\STATE $\hat\psi^{t+1}(s\vert\theta)\gets \frac{\sum_{\tau\in [t]}\sum_{i\in \mc N}\mathbbm{1}\left[s_i^\tau=s, \theta^\tau=\theta\right]}{t n p(\theta)}\quad\substack{\forall s\in S \\ \forall\theta\in\Theta}$ 
%		\FOR{$s,\theta\in S\times\Theta$}
%		\STATE $\hat\psi^{t+1}(s\vert\theta)\gets$ update as in Equation \eqref{eq:estsign}
%		\ENDFOR 
		\STATE Update estimators as in \cref{eq:estimatorfull}
		%\STATE $\varepsilon^{t+1}\gets \sqrt{\frac{\log (2T |S||\Theta|/\delta)}{2tn}}$
		%		\STATE $\varepsilon^{t+1}\gets$ update as in Equation \eqref{eq:eps}
		%		\STATE $\nu^{t+1}\gets2n |S||\Theta|\varepsilon_\delta^{t+1}$ 
		\ENDFOR
	\end{algorithmic}
\end{algorithm}
The procedure that we employ to solve the online problem in the full feedback setting is presented in Algorithm \ref{alg:full}. 
The algorithm takes as input the number of rounds $T$, the desired confidence level $\delta$ determining the probability with which the regret and IC violation bound hold, the utility functions $u_\ag$ and $u_\med$, and the prior $p$. At each round $t\in [T]$, the algorithm uses the information collected by $\med$ up to $t$ and selects a mechanism $\xivec^t$ that is optimal for linear optimization problem $\LP(\nu^t, \hat\rvec_\med^t, \hat\rvec_\ag^t)$, where $\nu^t = 2 |\C|\varepsilon_\delta^t$. 
%Notice that the parameters of the LP are chosen such that the utility vectors are computed replacing the signaling scheme $\psi_\ag$ with the estimated signaling scheme $\hat\psi^t$.
%by means of the estimated signaling scheme $\hat\psi^t$.
 Notice that the feasibility parameter $\nu^t$, which regulates the slackness of the linear program's incentive compatibility constraint, depends on the confidence bound $\varepsilon_\delta^t$. Intuitively, this is needed to compensate the uncertainty arising from the use of the estimated utility parameters $\hat\rvec_{\med}^t$, $\hat\rvec_{\ag}^t$ instead of the true parameters $\rvec_{\med}$, $\rvec_{\ag}$. This choice of the slackness parameter $\nu^t$ is fundamental in order to guarantee, with high probability, sublinear cumulative regret and IC violations. Indeed, when event $\E_\delta^\full$ holds, on the one hand it guarantees that the optimal IC mechanism $\xivec^\star$ is in the set of feasible mechanisms of $\LP(\nu^t, \hat\rvec_\med^t, \hat\rvec_\ag^{t})$ implying large receiver's utility, while, on the other hand, it ensures that $\xivec^t$ is $2\nu^t$-IC with respect to the real utility  vectors $\rvec_{\med}$, $\rvec_{\ag}$, implying a small violation of the IC constraints. These two properties can then be leveraged to obtain the following Theorem.

\begin{restatable}{theorem}{thfull}	\label{th:thfull}
	For any $\delta\in (0, 1)$, with probability at least $1-\delta$, Algorithm \ref{alg:full} guarantees 
	\begin{align*}
		R^T &= \mc O \left(|\C|\sqrt{8T\log(4T|\C||\A|/\delta)}\right), \\
		V^T &= \mc O \left(|\C|\sqrt{16T\log(4T|\C||\A|/\delta)}\right).
	\end{align*}
\end{restatable}
\section{Online Information Acquisition with Bandit Feedback}\label{sec:bandit}
Differently from the full feedback setting, where the receiver has access to complete observations of the signal profiles and states of nature, in the bandit feedback setting the receiver does not have access to this information and can only observe the utility $u_j(a^t, \theta^t)$ for $j\in\left\{\med, \ag\right\}$, the received signal profiles $\svec^t$, and the played actions $a^t$. This particular feedback introduces significant complications when trying to estimate the game parameters, thus requiring an adequate \emph{exploration} of the game upfront. The algorithm presented in this section dedicates the first set of rounds to the refinement of the estimates of the game parameters, and subsequently leverages the information collected in this first phase, in order to achieve the desired performances. 
In the first part of this section we introduce the estimators that are used by our procedure, while the second part of the Section is devoted to the presentation of the algorithm. Finally, we provide a lower bound showing that the regret and IC violation bounds attained by our algorithm are tight, thus suggesting the need for an accurate exploration of the game in the first rounds of the interaction. 

\subsection{Estimation of $\rvec_\med$ and $\rvec_\ag$}
In order to appropriately define the estimators of the game parameters, we require the receiver to play at each round $t\in [T]$ a deterministic mechanism\footnote{Let us point out that, given a mechanism $\xivec\in\Xi$, a deterministic mechanism $\pivec\sim\xivec$ can be efficiently sampled from $\xivec$, simply by iterating over all elements $c\in\C$ and by sampling an action according to the probability distribution specified by $\xi[c, \cdot]$.} $\pivec^t\in\Pi$.
While in the full feedback setting the receiver could leverage the observations concerning the sampled state of nature to implicitly estimate the signaling scheme $\psi_\ag$, under bandit feedback the information available to $\med$ is not sufficient to acquire knowledge on $\psi_\ag$. Instead, the observations of the utility values $u_\med^t$ and $u_\ag^t$, the signal profile $\svec^t$ and the action $a^t$, make it convenient to estimate directly the utility parameters $\rvec_\med$ and $\rvec_\ag$. In particular, using the information collected by $\med$ up to time $t\in [T]$, it is possible to define, for all $ j \in\left\{\med,\ag\right\}$, $c\in\C$, and $a\in\A$, the following estimator of the utility vectors:
\begin{equation}\label{eq:estvect}
	\hat r^{t+1}_j [c, a] \defeq \frac{1}{N^t[c, a]}\sum_{\substack{\tau\in [t]\\\pi^\tau[c, a]=1}} u_j(a^\tau, \theta^\tau)\mathbbm{1}\left[s^\tau\in c\right],
\end{equation}
where $N^t[c, a] \defeq \sum_{\tau\in [t]} \pi^\tau[c, a]$ is a counter that keeps track of how many times $\med$ selected a deterministic mechanism that prescribed to play action $a$ when receiving a signal profile belonging to class $c$. 
%Intuitively, averaging over all the samples in $N^t[c,a]$ allows our estimator to capture the probability of sampling a signal profile $\svec$ belonging to class $c$. More formally, this can be seen from the fact that $\mathbb{E}\left[\mathbbm{1}\left[\svec^\tau \in c\right]\right] = \sum_{\svec\in c} \mathbb{P}(\svec^\tau = \svec)$, which can be used to show that our estimators $\hat\rvec^t_j$ constitute unbiased estimators of the utility vectors $\rvec_j^{\psi_\ag}$. \ma{qui non si capisce molto. Ma non saprei come spiegarlo meglio}
%
%\color{blue}
\begin{restatable}{lemma}{lemmaestband}\label{le:lemmaestbandit}
	For all $t\in [T]$, for all $i\in\left\{\med, \ag\right\}$, $\hat\rvec^t_j$ is an unbiased estimator of $\rvec_i$. %Formally, 
%	\[
%		\mathbb E\left[\hat\rvec^t_j\right] = \rvec_j^{\psi_\ag} \quad\forall j\in\left\{\med, \ag\right\},\forall t\in [T].
%	\]
\end{restatable}
%\color{black}
Given $\delta\in (0,1)$, let us define the event  $\E_\delta^\bandit$ such that
\[
	\begin{split}
		\E_\delta^\bandit\defeq\left\{|\hat r_i^t[c, a]\right.& - r_i[c, a]| \leq \eta_\delta^t[c,a]\\ 
		&\left. \forall t\in[T], \forall c\in\C,\forall a\in\A  \forall i\in\left\{\med, \ag\right\} \right\}, 
	\end{split}
\]
where $\etavec_\delta^t\in\Reals^{|\C|\times|\A|}$ is the vector such that, for any ${\delta\in (0,1)}$ and for all $t\in [T]$, 
\begin{equation}\label{eq:confid}
	\eta_\delta^t[c,a] :=\sqrt{\frac{\log( 8T|\C||\A|/\delta)}{2 N^t[c, a]}}.
\end{equation}

Also in this case, by means of standard concentration arguments, it is possible to derive a lower bound on the probability with which event $\E_\delta^\bandit$ is verified, thus yielding an high confidence region for the utility vectors $\rvec_\med$ and $\rvec_\ag$: 
\begin{restatable}{lemma}{lemmaconcbandit}\label{le:lemmaconcbandit}
	For any $\delta\in (0,1)$, the following holds: 
	\[
		\mathbb{P}\left(\E_\delta^\bandit\right)\geq 1-\frac{\delta}{2}.
	\]
\end{restatable}
\subsection{Algorithm}
\begin{algorithm}
	\caption{Algorithm for the bandit feedback setting}\label{alg:bandit}
	\begin{algorithmic}
		\REQUIRE $T, E\in \mathbb{N}_{>0}$, $\delta\in (0,1)$
		\STATE \textbf{Initialize:}
		\STATE $\quad\hat r_j^1[c,a] \gets 0\quad\forall c\in\C,\,\forall a\in\A,\,\forall j\in\left\{\med,\ag\right\}$
		\STATE $\quad\eta_\delta^1[c,a]\gets\infty,\,\quad N^0[c, a]\gets 0\quad\forall c\in\C,\,\forall a\in\A$
		\STATE $\quad \nu\gets 2|\C|\sqrt{\frac{\log( 8T|\C||\A|/\delta)}{2 E}} $
		\FOR{$t\in [T]$}
		
		\IF[Exploration rounds]{$t \leq |\A|E$} 
		\STATE $a\gets \arg\min_{a\in\A} \sum_{c\in\C}N^t[c,a]$
		\STATE $\xi^t[c,a^\prime]\gets \mathbbm{1}[a^\prime = a]\quad\forall c\in\C\,\forall a\in\A$.
		\ELSE[Optimization rounds]
		\STATE $\xivec^t\gets$ optimal solution to $\LP\left(\nu, \hat\rvec^t_\med + \etavec_\delta^t, \hat\rvec_\ag^t \right)$
		\ENDIF
		\STATE \LComment{Interaction with the senders}
		\STATE $\pivec^t\gets $ sample deterministic mechanism from $\xivec^t$
		\STATE Use $\pivec^t$ and observe $\svec^t$, $a^t$, $u_\med^t$, $u_\ag^t$
		\STATE Update Estimators
		%		\STATE \LComment{Estimators update}
		%		\FOR{$c, a\in\C\times \A$}
		%		\STATE $N^{t}[c,a]\gets N^{t-1}[c,a] + \pi^t[c,a]$		
		%		%\STATE $\hat r_j^{t+1}[c, a]\gets\frac{\sum_{\tau\in [t]} u_j^\tau \pi^\tau[c, a]\mathbbm{1}[\svec^\tau\in c, ]}{N^{t}[c,a]} \quad\forall j\in\left\{\med,\ag\right\}$
		%		\STATE $\hat r_j^{t+1}[c, a]\gets$ update as in Equation \eqref{eq:estvect} $\,\forall j\in\left\{\med,\ag\right\}$
		%		\STATE $\eta_\delta^t[c,a] \gets$ update as in Equation \eqref{eq:confid}
		%		\ENDFOR		
		\ENDFOR
	\end{algorithmic}
\end{algorithm}
The algorithm for the bandit feedback setting is described in Algorithm \ref{alg:bandit}. 
The algorithm, which is formulated as a two-phases procedure, takes as input the time horizon $T$, the desired confidence level on the regret and IC violation bounds $\delta$, and the number $E$ of rounds that will be devoted to the exploration of each action. More in detail, in the first phase, which is called \emph{exploration phase}, the objective of the receiver is to select mechanisms $\xivec^t$ that guarantee to obtain sufficiently accurate estimations of the game parameters. To this extent, the exploration phase is designed such that, for each action $a\in\A$, $\med$ plays for exactly $E$ rounds a deterministic mechanism $\pivec$ that prescribes to play action $a$ for each class $c\in\C$, \emph{i.e.,} a $\pivec$ such that $\pi[c, a] = 1$, $\forall c\in\C$. Intuitively, by imposing that each couple $c, a\in\C\times\A$ is explored for an adequate number of rounds within the exploration phase, we are in fact guaranteeing an upper bound $\nu = 2|\C|\sqrt{\log( 8T|\C||\A|/\delta)/2E}$ on the confidence parameters $\etavec_\delta^t$ at subsequent rounds $t$. Such an upper bound will then prove to be crucial in order to obtain the desired performances. Furthermore, it is worth pointing out that such deterministic mechanisms do \emph{not} contribute to the cumulative IC violation, since, trivially, selecting the same action for any possible class $c$ constitutes an IC mechanism. 

After the first $|\A|E$ rounds of exploration, the algorithm enters its second phase, which is called \emph{optimization phase}. In this second set of rounds, the receiver aims at exploiting the information collected in order to select mechanisms that contribute to minimize the cumulative regret and IC violation. In particular, at each round $t > |\A|E$, the estimators $\hat\rvec_\med^t$ and  $\hat\rvec_\ag^t$, together with the confidence bounds $\etavec_\delta^t$ and $\nu$, are used in the strategy selection procedure as input parameters to the linear optimization problem $\LP\left(\nu, \hat\rvec^t_\med + \etavec_\delta^t, \hat\rvec_\ag^t \right)$. More precisely, the linear program is instantiated so that the objective is exactly the optimization of the expected utility computed with respect to the upper confidence bound of the parameter $\rvec_\med$ and the IC constraint is specified using the estimator $\hat\rvec_\ag^t$ as the utility parameter, with constraint slackness regulated by the parameter $\nu$. 
These design choices have a twofold effect: (i) on the one hand they ensure that the optimal IC mechanism $\xivec^\star$ is in the feasibility set of the linear program, (ii) while on the other hand they provide guarantees on the approximate incentive compatibility of mechanisms $\xivec^t$ chosen at each round $t > |\A|E$. In particular, while (i), together with the choice of a UCB-like objective, yields optimal regret bounds, thanks to (ii) it is possible to recover optimal IC violation for Algorithm \ref{alg:bandit}. 
The formal guarantees of the algorithm are stated in the following Theorem:

\begin{restatable}{theorem}{thbandit}\label{th:bandit}
	For any $\delta\in (0,1)$ and for any $E\in [T]$, with probability at least $1-\delta$, Algorithm \ref{alg:bandit} guarantees
	\begin{align*}
		R^T &= \mc O\left(|\A|E + |\C||\A| \sqrt{32T\log( 8T|\C||\A|/\delta)}\right),\\
		V^T &= \mc O\left(T|\C|\sqrt{8\log( 8T|\C||\A|/\delta)/E}\right).
	\end{align*}
\end{restatable}
When the number of exploration rounds is set as ${E = \lfloor T^{\alpha}\rfloor}$ with $\alpha \in [1/2, 1]$, \cref{th:bandit} implies that the cumulative regret of Algorithm \ref{alg:bandit} is $R^T = \tilde{\mc O}(T^{\alpha})$, while the cumulative IC violation is $V^T=\tilde{\mc O}(T^{1-\alpha/2})$
Clearly, the choice of the number of exploration rounds introduces a fundamental trade-off between minimization of the cumulative regret and minimization of the cumulative IC violation. This trade-off is inherently related to the structure of the bandit feedback setting itself. Indeed, as we show in the remaining part of this section, it is not possible to improve over the regret and violation bounds that are attained by Algorithm \ref{alg:bandit}.

\subsection{Lower Bound}
Before concluding the analysis of the bandit feedback setting, we derive a lower bound showing that the regret and IC violation bounds achieved by Algorithm \ref{alg:bandit} are tight. 

Intuitively, the lower bound captures the need for exploration that characterizes the bandit feedback setting, fundamentally distinguishing it from the full feedback setting. Indeed, while, in the latter, the feedback received at each round $t$ is completely informative and can be used to improve the receiver's knowledge of the game independently on the chosen mechanism $\xivec^t$, in the former setting the feedback is strongly related to the mechanisms $\xivec^t$ selected by $\med$. This aspect makes it necessary to guarantee a suitable level of exploration, in fact -- as already observed above -- bringing up an actual trade-off between minimization of cumulative regret and minimization of cumulative IC deviation. The following Theorem formalizes the above ideas.

\begin{restatable}{theorem}{lowerbound}\label{th:lowerbound}
	For any $\alpha\in[1/2, 1]$, there exists $\delta\in (0,1)$ such that no algorithm can achieve both $R^T = o(T^\alpha)$ and $V^T = o(T^{1-\alpha/2})$ with probability greater than $1-\delta$.
\end{restatable}

As we can observe from Theorem \ref{th:lowerbound}, for $E=\lfloor T^\alpha\rfloor$ and $\alpha\in [1/2,1]$, the regret and IC violation bounds attained by Algorithm \ref{alg:bandit} are tight. 

\section*{Acknowledgements}
This paper is supported by FAIR (Future Artificial Intelligence Research) project, funded by the NextGenerationEU program within the PNRR-PE-AI scheme (M4C2, Investment 1.3, Line on Artificial Intelligence).

%\input{src/concl}

%\clearpage
\bibliographystyle{icml2023}
\bibliography{bibliography}

\clearpage
\onecolumn
\appendix
\section{Proofs Omitted from Section \ref{sec:prelim}}
\propgreen*
\begin{proof}
	First, let us point out that given any mechanism $\xvec\in\X$, the utilities, $U_i(\xvec)$ and $U_i^\varphi(\xvec)$ of each deviation $\varphi$ can be expressed as the following linear functions: 
	\[
		\begin{split}
			U_i(\xvec) &= \xvec^\top \dvec_i, \\
			U_i^\varphi(\xvec) &= \xvec^\top \Mmat^\varphi_i \dvec_i,
		\end{split}
		\quad\forall i\in\mc N,
	\]
	where $\dvec_i\in\Reals^{|\Si||\A|}$ is the vector such that
	\[
		d_i[\svec, a] = \sum_{\theta\in\Theta}p(\theta) \psi_{\ag}(\svec\vert\theta) u_i(a, \theta)\quad\forall\svec\in\Si,\,\forall a\in\A,
	\]
	and $\Mmat^\varphi_i\in\Reals^{|\Si||\A|\times|\Si||\A|}$ is the matrix such that
	\[
		M[(\svec, a), (\svec^\prime, a^\prime)] = \begin{cases}
			1 & \text{if } \svec_{-i} = \svec^\prime_{-i},\, \varphi(s_i) = s_i^\prime,\, a = a^\prime \\
			0& \text{otherwise.}
		\end{cases}
	\]
	For any $s,s^\prime\in S$, let $\Mmat^{s\to s^\prime}$ be the matrix associated to deviation function  $\varphi^{s\to s^\prime}\in\overline\Phi(s)$ such that $\varphi^{s\to s^\prime}(s) =s^\prime$. Furthermore, let $\Mmat^I$ be the matrix associated to the identity deviation function $\varphi^I$, \emph{i.e.,} the deviation function such that $\varphi^I(s)=s$, for any $s\in S$. Then, for any $\varphi\in\Phi^{inter}$, we can write the following:
	\[
		\Mmat^{\varphi} = \sum_{s\in S}\left[\Mmat^{s \to\varphi(s)}\right] + (1 - |S|) \Mmat^I.
	\]
	Now, let $\xvec$ be an $\varepsilon$-IC mechanism with respect to deviation functions $\overline\Phi$. It holds that for each $\varphi\in\Phi^{inter}$,
	\begin{align*}
		U_i^\varphi(\xvec) - U_i(\xvec) &= \xvec^\top\Mmat_i^\varphi \dvec_i - \xvec^\top\dvec_i \\
		&= \sum_{s\in S}\left[\xvec^\top \Mmat^{s\to\varphi(s)}\dvec_i\right] + \xvec^\top \Mmat^{I}\dvec_i - |S| \xvec^\top \Mmat^{I}\dvec_i - \xvec^\top\dvec_i \\
		&= \sum_{s\in S}\left[\xvec^\top \Mmat^{s\to\varphi(s)}\dvec_i - \xvec^\top \dvec_i\right] \\
		&\leq |S|\varepsilon,
	\end{align*}
	where the third equation follows from the fact that $\xvec^\top \Mmat^I \dvec_i = \xvec^\top \dvec_i$, and the last inequality follows from the fact that $\xvec$ is $\varepsilon$-IC with respect to deviation functions $\overline\Phi$.
	
	This concludes the proof.	
\end{proof}
\section{Proofs Omitted from Section \ref{sec:compact}}

\subsection{Proof of Theorem \ref{th:thsymmetry}}
Before proving Theorem \ref{th:thsymmetry}, let us introduce some preliminary notation and results that will be needed in the proof. 

Let $\Lambda$ be the set of permutations of the elements in the set $[n]$. For each $\lambdavec\in\Lambda$, let $f_{\lambdavec}:\Si \to \Si$ be the permutation function that given a signal profile $\svec\in\Si$ returns a signal profile $\svec^\prime\in\Si$ obtained changing the order of signals in $\svec$ according to the permutation $\lambdavec$. Formally, 
\begin{definition}
	For any $\lambdavec\in\Lambda$, the permutation function $f_{\lambdavec}:\Si\to\Si$ is defined as the bijective function such that $\forall \svec\in\Si$, $f_{\lambdavec}(\svec) = \svec^\prime\in\Si$, where 
	\[
	s^\prime_i = s_{\lambda[i]}\quad\forall i\in\mc N.
	\]
\end{definition}
\begin{definition}
	For any $\lambdavec\in\Lambda$, the inverse permutation function $f_{\lambdavec}^{-1}:\Si\to\Si$ is defined as the bijective function such that $\forall \svec\in\Si$, $f_{\lambdavec}^{-1}(\svec) = \svec^\prime\in\Si$, where 
	\[
	s^\prime_{\lambda[i]} = s_{i}\quad\forall i\in\mc N.
	\]
\end{definition}
For notational convenience, for any $i\in\mc N$ and $\varphi\in\Phi$, we introduce the vectors $\dvec_\med, \dvec_i, \dvec_i^\varphi\in\Reals^{|\Si||\A|}$, where: 
\begin{align*}
	d_i[\svec, a] &= \sum_{\theta\in\Theta} p(\theta) u_i(a, \theta) \psi_\ag(\svec\vert\theta) \quad \forall i\in\mc N\cup\left\{\med \right\}\\
	d_i^\varphi[\svec, a] &= \sum_{\substack{\svec^\prime\in\Si \\ \svec_{-i} = \svec^\prime_{-i}\\ \varphi(s_i^\prime) = s_i}} d_i[\svec^\prime, a] = \sum_{\theta\in\Theta} p(\theta) u_i(a, \theta) \sum_{\substack{\svec^\prime\in\Si \\ \svec_{-i} = \svec^\prime_{-i}\\ \varphi(s_i^\prime) = s_i}} \psi_\ag(\svec\vert\theta)\quad\forall i\in\mc N.
\end{align*}
Then, exploiting the above formulations, we can rewrite the expected utilities as: 
\begin{align*}
	U_i(\xvec) &= \sum_{\substack{\svec\in\Si\\ a\in\A}} x[\svec, a] d_i[\svec, a] \quad\forall i\in\mc N\cup\left\{\med\right\}\\
	U_i^\varphi(\svec) &= \sum_{\substack{\svec\in\Si\\ a\in\A}} x[\svec, a] d_i^\varphi[\svec, a] \quad\forall i\in\mc N,\,\forall \varphi\in\Phi.
\end{align*} 

We can state the following auxiliary result:
\begin{lemma}
	For any $\varepsilon\geq 0$, let $\xvec^\star\in \argmax_{\xvec\in\X_\varepsilon} U_\med\left(\xvec\right)$. Define, for each permutation $\lambdavec\in\Lambda$, the mechanism $\xvec_{\lambdavec}\in\X$ as the mechanism such that, $\forall \svec\in\Si,\,\forall a \in\A$, $\xvec^\star[\svec, a] = \xvec_{\lambdavec}[f_{\lambdavec}(\svec), a]$. Then, the following holds: 
	\[
	\xvec_{\lambdavec} \in \argmax_{\xvec\in\X_\varepsilon}U_\med\left(\xvec\right).
	\]
	\label{le:lemmasymmetryaux}
\end{lemma}
\begin{proof}
	In order to show that $\xvec_{\lambdavec} \in \argmax_{\xvec\in\X_\varepsilon}U_\med\left(\xvec\right)$, we first show that $U_\med\left(\xvec^\star\right) = U_\med\left(\xvec_{\lambdavec}\right)$ and then that $\xvec_{\lambdavec}\in\X_\varepsilon$.
	\paragraph{Objective function.}
	Note that $\forall\theta\in\Theta$, $\forall \svec,\svec^\prime\in\Si$ such that $\svec\bowtie\svec^\prime$, it holds that 
	\[
	\psi_\ag(\svec\vert\theta) = \psi_\ag(\svec^\prime\vert\theta).
	\]
	Hence, for any $i\in\mc N\cup\left\{\med\right\}$ and $\svec,\svec^\prime\in\Si$ such that $\svec\bowtie\svec^\prime$
	\begin{equation}
		d_i[\svec, a] = \sum_{\theta\in\Theta} p(\theta) u_i(a, \theta) \psi_\ag(\svec\vert\theta) = \sum_{\theta\in\Theta} p(\theta) u_i(a, \theta) \psi_\ag(\svec^\prime\vert\theta) = d_i[\svec^\prime, a]\label{eq:equ_obj}
	\end{equation}

	Furthermore, since the permuation function preserves the elements present in the signal profile and only changes their order, $\forall \svec\in\Si$, $\forall \lambdavec\in\Lambda$, we have that $f_{\lambdavec}(\svec)\bowtie\svec$.
	Thus, $\forall i\in\mc N\cup\left\{\med\right\}$, we can write the following equations: 
	\begin{align}
		U_i(\xvec_{\lambdavec}) &= \sum_{\substack{\svec\in\Si\\ a\in\A}} x_{\lambdavec}[\svec, a] d_i[\svec, a] \notag \\
		&= \sum_{\substack{\svec\in\Si\\ a\in\A}} x^\star[f^{-1}_{\lambdavec}(\svec), a] d_i[\svec, a]\notag \\
		&= \sum_{\substack{\svec\in\Si\\ a\in\A}} x^\star[\svec, a] d_i[f_{\lambdavec}(\svec), a] \notag\\
		&= \sum_{\substack{\svec\in\Si\\ a\in\A}} x^\star[\svec, a] d_i[\svec, a] \notag\\
		&= U_i(\xvec^\star), \label{eq:def_lhs}
	\end{align}
	where the first equation follows from the definition of $\xvec_{\lambdavec}$, the second and third equation follow from the fact that the function $f_{\lambdavec}$ is bijective and the last equation follows from \Cref{eq:equ_obj}. This concludes the first part of the proof.
	
	\paragraph{Feasibility.}

	Let $\Si_{-j} = \times_{i\in\mc N\setminus\left\{j\right\}} \Si_i$.	Notice that $\forall \svec\in\Si$, $\forall\lambdavec\in\Lambda$,
	\[
		\left\{\svec^\prime_{-i}\in\Si_{-i}\mid \svec^\prime_{-i} = f_{\lambdavec}(\svec)_{-i}\right\} = \left\{\svec^\prime_{-\lambda[i]}\in\Si_{-\lambda[i]}\mid \svec^\prime_{-\lambda[i]} = \svec_{-\lambda[i]}\right\}.
	\] 
	Then, given any signal profile $\svec\in\Si$ and any permuation $\lambdavec\in\Lambda$, we have that $\forall i\in\mc N,\,\forall \varphi\in\Phi$
	\begin{align}
		d^\varphi_i[f_{\lambdavec}(\svec), a] &= \sum_{\substack{\svec^\prime\in\Si \\ f_{\lambdavec}(\svec)_{-i} = \svec^\prime_{-i}\\ f_{\lambdavec}(\svec)_i = \varphi(s_i^\prime)}} d_i[\svec^\prime, a] \notag \\
		&= \sum_{\substack{s^\prime_{\lambda[i]}\in S\notag \\ \varphi(s^\prime_{\lambda[i]}) = s_{\lambda[i]}}} \sum_{\substack{\svec^\prime_{-\lambda[i]}\in\Si_{-\lambda[i]}\notag\\ \svec^\prime_{-\lambda[i]} = \svec_{-\lambda[i]} }} d_i[(s^\prime_{-\lambda[i]}, \svec^\prime_{-\lambda[i]}), a] \notag\\
		&= \sum_{\substack{\svec^\prime\in\Si \\ \svec_{-\lambda[i]} = \svec^\prime_{-\lambda[i]}\\ s_{\lambda[i]} = \varphi(s_{\lambda[i]}^\prime)}} d_i[\svec^\prime, a]  \notag\\
		&= d_{\lambda[i]}^\varphi[\svec, a] \label{eq:equ_rhs}.
	\end{align}
%	\color{blue}
%	Notice that $\forall \svec\in\Si$, $\forall\lambdavec\in\Lambda$, $f_{\lambdavec}(\svec)_{-i} = \svec_{-\lambda[i]}$
%	Given any signal profile $\svec\in\Si$ and any permuation $\lambdavec\in\Lambda$, we have that $\forall i\in\mc N,\,\forall \varphi\in\Phi$
%	\[
%	\sum_{\substack{\svec^\prime\in\Si \\ f_{\lambdavec}(\svec)_{-i} = \svec^\prime_{-i}\\ \varphi(s_i^\prime) = f_{\lambdavec}(\svec)_i}} \psi_\ag(\svec^\prime\vert\theta) = \sum_{\substack{\svec^\prime\in\Si \\ \svec_{-\lambda[i]} = \svec^\prime_{-\lambda[i]}\\ \varphi(s_{\lambda[i]}^\prime) = s_{\lambda[i]}}} \psi_\ag(\svec^\prime\vert\theta).
%	\]
%	
%	Thus, similarly to the previous case, $\forall \svec\in\Si,\, a\in\A,\, i\in\mc N,\, \varphi\in\Phi$, the following holds
%	\begin{equation}
%		d^\varphi_i[f_{\lambdavec}(\svec), a] = \sum_{\theta\in\Theta}p(\theta)u_i(a, \theta) \sum_{\substack{\svec^\prime\in\Si \\ f_{\lambdavec}(\svec)_{-i} = \svec^\prime_{-i}\\ \varphi(s_i^\prime) = f_{\lambdavec}(\svec)_i}} \psi_\ag(\svec^\prime\vert\theta) = \sum_{\theta\in\Theta}p(\theta)u_i(a, \theta) \sum_{\substack{\svec^\prime\in\Si \\ \svec_{-\lambda[i]} = \svec^\prime_{-\lambda[i]}\\ \varphi(s_{\lambda[i]}^\prime) = s_{\lambda[i]}}} \psi_\ag(\svec^\prime\vert\theta) = d^\varphi_{\lambda[i]}[\svec, a]. 
%	\end{equation}
%	\color{black}
	
	Then, it holds that 
	\begin{align}
		U_i^\varphi(\xvec_{\lambdavec}) &= \sum_{\substack{\svec\in\Si\\ a\in\A}} x_{\lambdavec}[\svec, a] d_i^\varphi[\svec, a]\notag \\
		&= \sum_{\substack{\svec\in\Si\\ a\in\A}} x^\star[f^{-1}_{\lambdavec}(\svec), a] d_i^\varphi[\svec, a] \notag\\
		&= \sum_{\substack{\svec\in\Si\\ a\in\A}} x^\star[\svec, a] d_i^\varphi[f_{\lambdavec}(\svec), a] \notag\\
		&= \sum_{\substack{\svec\in\Si\\ a\in\A}} x^\star[\svec, a] d_{\lambda[i]}^\varphi[\svec, a] \notag\\
		&= U_{\lambda[i]}^\varphi(\xvec^\star), \label{eq:def_rhs}
	\end{align}
	where the first equation follows from the definition of $\xvec_{\lambdavec}$, the second and third equations follow from the fact that the permutation function is bijective and the last equation follows from \Cref{eq:equ_rhs}.
	
	Hence, we can conclude that, $\forall i\in\mc N$, $\forall \varphi\in\Phi$ and $\forall \lambdavec\in\Lambda$, 
	\begin{align*}
		U_i^\varphi(\xvec_{\lambdavec}) - U_i(\xvec_{\lambdavec}) = U_{\lambda[i]}^\varphi(\xvec^\star) - U_i(\xvec^\star) \leq \varepsilon
		% \\  
		%&= U_{\lambda[i]}(\xvec^\star, \phi_\ag) - U_{\lambda[i]}(\xvec^\star, \phi_\ag) \\
		%&\leq \varepsilon,
	\end{align*}
	where the first equation follows from Equations \eqref{eq:def_lhs} and \eqref{eq:def_rhs}, and the last inequality follows from the fact that $\xvec^\star\in\X_\varepsilon$. 
	
	This concludes the proof.
\end{proof}

We are now ready to prove Theorem \ref{th:thsymmetry}

\thsymmetry*
\begin{proof}
	Let $\xvec^\star\in\X\in\argmax_{\xvec\in\X_\varepsilon} U_\med(\xvec)$. Define, for each permutation $\lambdavec\in\Lambda$, the permutation strategy $\xvec_{\lambdavec}$ such that $\forall\svec\in\Si,\,\forall a\in\A$, $x^\star[\svec, a] = x_{\lambdavec}[f_{\lambdavec}(\svec), a]$. The set of all possible permutation strategies is the following 
	\[
	\X^\Lambda = \left\{ \xvec_{\lambdavec}\mid \lambdavec\in\Lambda \right\}.
	\]
	Let $\overline\xvec$ be such that:
	\[
	\overline\xvec = \frac{1}{|\Lambda|}\sum_{\lambdavec\in\lambda}\xvec_{\lambdavec}.
	\]
	By convexity and by Lemma \ref{le:lemmasymmetryaux}, $\overline\xvec\in\argmax_{\xvec\in\X_\varepsilon} U_\med(\xvec)$. In order to conclude the proof, we need to show that $\forall \svec, \svec^\prime\in\Si$ such that $\svec\bowtie\svec^\prime$, it holds that 
	\[
	\overline{x}[\svec, a] = \overline{x}[\svec^\prime, a]\quad\forall a\in\A.
	\]
	Fix any $\svec, \svec^\prime\in\Si$ such that $\svec\bowtie\svec^\prime$. By definition of permutation, it holds that: 
	\[
	\left\{f_{\lambdavec}(\svec)\mid\lambdavec\in\Lambda\right\} = \left\{f_{\lambdavec}^{-1}(\svec)\mid\lambdavec\in\Lambda\right\} = \left\{f_{\lambdavec}^{-1}(\svec^\prime)\mid\lambdavec\in\Lambda\right\} = \left\{f_{\lambdavec}(\svec^\prime)\mid\lambdavec\in\Lambda\right\} = \left\{\svec^{\prime\prime}\in\Si\mid\svec^{\prime\prime}\bowtie\svec\right\}.
	\]
	Hence, $\forall \svec,\svec^\prime\in\Si$ such that $\svec\bowtie\svec^\prime$, and $\forall a\in\A$,
	\begin{align*}
		\overline x[\svec, a] &= \frac{1}{|\Lambda|}\sum_{\lambdavec\in\Lambda} x_{\lambdavec}[\svec, a] \\
		&= \frac{1}{|\Lambda|}\sum_{\lambdavec\in\Lambda} x^\star[f^{-1}_{\lambdavec}(\svec), a] \\
		&= \frac{1}{|\Lambda|}\sum_{\lambdavec\in\Lambda} x^\star[f^{-1}_{\lambdavec}(\svec^\prime), a] \\ 
		&= \frac{1}{|\Lambda|}\sum_{\lambdavec\in\Lambda} x_{\lambdavec}[\svec^\prime, a] \\
		&= \overline x[\svec^\prime, a].
	\end{align*}
	This concludes the proof.
\end{proof}

\subsection{Additional Proofs}
\lemmapoly*
\begin{proof}
	To prove that the polytope $\Xi$ admits a polynomially sized representation we provide an upper bound on the number of variables and constraints needed to describe $\Xi$.
	\paragraph{Bound on the number of variables.} By definition, we have that $\Xi\subset \Reals^{|\C||\A|}_{\geq 0}$. Hence, in order to bound the dimension of $\Xi$, we need to bound the number of elements in $\C$. In particular, we have that each element $c\in\C$ can be uniquely identified by a vector $\cvec\in\Naturals^{|S|}$ in which the $i$-th entry $c[i]$ defines the number of occurrences of signal $s_i\in S$. Thus, we can write the following:
	\[
	c[i] \leq n \quad \forall i\in [|S|].
	\]
	As a consequence of this, it is possible to upper bound the maximum number of elements in $\C$ as 
	\[
	|\C| < n^{|S|},
	\]
	hence, the number $m_1$ of variables is
	\[
	m_1 = |\C|\cdot |\A| < |\A| n^{|S|}.
	\]
	
	\paragraph{Bound on the number of constraints.} The constraints that characterize the polytope $\Xi$ are exactly the simplex constraints associated to each $c\in\C$. Formally:
	\begin{align}
		\xi[c, a] &\geq 0\quad\forall c\in\C,\,\forall a\in\A,\notag\\
		\sum_{a\in\A} \xi[c, a] &= 1 \quad\forall c\in\C \notag.
	\end{align} 
	Thus, the number $m_2$ of constraints is
	\[
	m_2 = |\C|\cdot |\A| + |\C| < n^{|S|}\left(|\A| + 1\right).
	\]
	The Lemma follows from the fact that, by Assumption \ref{ass:const}, $|S|$ is an absolute constant.
\end{proof}

\lemmadev*
\begin{proof}
	For each $i\in\mc N$ and for each deviation function $\varphi\in\Phi$, from the definition of $U_i^\varphi$ it follows that
	\begin{align*}
		U_i^\varphi(\xivec) &= \sum_{\substack{\svec\in\Si\\a\in\A}} \xi[c((\varphi(s_i), \svec_{-i})), a] \sum_{\theta\in\Theta} p(\theta)u_\ag(a, \theta)\psi_\ag(\svec\vert\theta) \\
		&= \sum_{\substack{c\in\C\\ a\in\A}} \sum_{\svec\in c} \xi[c((\varphi(s_i), \svec_{-i})), a] \sum_{\theta\in\Theta} p(\theta)u_\ag(a, \theta)\psi_\ag(\svec\vert\theta) \\
		&=  \sum_{\substack{c\in\C\\ a\in\A}} \sum_{\theta\in\Theta} p(\theta)u_\ag(a, \theta) \sum_{\svec\in c} \xi[c((\varphi(s_i), \svec_{-i})), a]  \psi_\ag(\svec\vert\theta).
	\end{align*}
	Let us recall that $\forall\theta\in\Theta$ and $\forall \svec\bowtie\svec^\prime$, $\psi_\ag(\svec\vert\theta) = \psi_\ag(\svec^\prime\vert\theta)$. Hence, since $\forall c\in\C$, $\svec,\svec^\prime\in c$ if and only if $\svec\bowtie\svec^\prime$, with a slight abuse of notation, we can write $\psi_\ag(\svec\vert\theta) = \psi_\ag(c(\svec)\vert\theta)$. Thus, 
	\begin{align}
		U_i^\varphi(\xivec) &= \sum_{\substack{c\in\C\\ a\in\A}}\sum_{\theta\in\Theta}p(\theta)u_\ag(a, \theta)\psi_\ag(c\vert\theta) \sum_{\svec\in c}\xi[c((\varphi(s_i), \svec_{-i})), a]. \label{eq:devdef}
	\end{align}
	By definition of $\rvec_\ag$, we have that $\forall c\in\C$, $\forall a\in\A$,
	\begin{align*}
		r_\ag[c, a] &= \sum_{\svec\in c} \sum_{\theta\in\Theta} p(\theta)u_\ag(a, \theta)\psi_\ag(\svec\vert\theta) \\
		&= \sum_{\theta\in\Theta} p(\theta)u_\ag(a, \theta) \sum_{\svec\in c}\psi_\ag(\svec\vert\theta) \\
		&= |c| \sum_{\theta\in\Theta}p(\theta)u_\ag(a, \theta)\psi_\ag(c\vert\theta).
	\end{align*}
	Substituting into Equation~\eqref{eq:devdef}, we get
	\begin{align}
		U_i^\varphi(\xivec) = \sum_{\substack{c\in\C\\ a\in\A}} r_\ag[c, a] \frac{\sum_{\svec\in c}\xi[c((\varphi(s_i), \svec_{-i})), a]}{|c|}.\label{eq:devdef1}
	\end{align}
	Now, notice that the term $\frac{1}{|c|}\sum_{\svec\in c}\xi[c((\varphi(s_i), \svec_{-i})), a]$ can be rewritten as
	\begin{align}
		\frac{1}{|c|}\sum_{\svec\in c}\xi[c((\varphi(s_i), \svec_{-i})), a] &= \sum_{c^\prime\in \C} \xi[c^\prime, a] \sum_{\svec\in c} \frac{\mathbbm{1}[(\varphi(s_i), \svec_{-i})\in c^\prime]}{|c|} \notag\\
		&= \sum_{c^\prime\in\C}\xi[c^\prime, a] A_i^\varphi[(c^\prime, a), (c, a)]\notag \\
		&= \sum_{\substack{c^\prime\in\C\\ a^\prime\in\A}}\xi[c^\prime, a^\prime] A_i^\varphi[(c^\prime, a^\prime), (c, a)],\label{eq:devdef2}
	\end{align} 
	where the last equation follows from the fact that $A[(c^\prime, a^\prime), (c, a)] = 0$, $\forall a^\prime\neq a$.
	
	Hence, by plugging Equation~\eqref{eq:devdef2} into Equation~\eqref{eq:devdef1}, it follows that
	\begin{align*}
		U_i^\varphi(\xivec) &= \sum_{\substack{c\in\C\\ a\in\A}} r_\ag[c, a] \sum_{\substack{c^\prime\in\C\\ a^\prime\in\A}}\xi[c^\prime, a^\prime] A_i^\varphi[(c^\prime, a^\prime), (c, a)] \\
		&= \xivec^\top \Amat_i^\varphi \rvec_\ag,
	\end{align*}
	which gives the result.
\end{proof}

\lemmasymmdev*
\begin{proof}
	Notice that each class $c\in\C$ contains a set of symmetric signal profiles. Hence, we can define the function $\lambda: c\to c$ such that $\forall \svec\in c$, $\lambda(\svec) =\svec^\prime$, where $s_i=s^\prime_j$, $s_j = s^\prime_i$ and $s_k = s^\prime_k$, $\forall k\neq i, j$ and obtain the following: 
	\begin{align*}
		\frac{\sum\limits_{s\in c} \mathbbm{1}[(\varphi(s_i), \svec_{-i}) \in c^\prime]}{|c|} &= \frac{\sum\limits_{s\in c} \mathbbm{1}[(\varphi(\lambda(\svec)_j), \lambda(\svec)_{-j}) \in c^\prime]}{|c|} \\ 
		&= \frac{\sum\limits_{s^\prime\in c} \mathbbm{1}[(\varphi(s^\prime_j), s^\prime_{-j}) \in c^\prime]}{|c|},
	\end{align*}
	which gives the result.
\end{proof}

\propeff*
\begin{proof}
	As a preliminary step towards the proof of this result, we first show how to efficiently determine the number of signal profiles belonging to a given class $c\in\C$, and then proceed to bounding the complexity of specifying the utility vectors $\rvec_\med$, $\rvec_\ag$ and the deviation matrix $\Amat_{\ag}^\varphi$. 
	
	Let us remark that, similarly to the proof of Lemma~\ref{le:lemmapoly}, each element $c\in\C$ can be uniquely identified by a vector $\cvec\in\Naturals^{S}$ in which the entry $c[s]$ defines the number of occurrences of signal ${s \in S}$. Hence, the number of different signal profiles in a given class $c$ can be computed via the multinomial coefficient
	\begin{equation}\label{eq:classcard}
		|c| = \binom{n}{\cvec} = \frac{n !}{\prod_{s \in S} c[s_\ag]!}.
	\end{equation}
	
	\paragraph{Computation of $\rvec_\med$ and $\rvec_\ag$.}
	Notice that, by definition, each $c\in\C$ contains signal profiles that are symmetric between them. Recalling the definition of symmetric signal profile, we can write the following: 
	\[
		\psi(\svec\vert\theta) = \psi(\svec^\prime\vert\theta)\quad\forall\theta\in\Theta,\,\forall\svec\bowtie\svec^\prime.
	\]
	Thus, with a slight abuse of notation, since all signal profiles belonging to the same class share the same sampling probability, we can define $\psi(c\vert\theta) \defeq \psi(\svec\vert\theta)$, $\forall\theta\in\Theta$, $\forall \svec\in c$. Hence, for each $j\in\left\{\med,\ag\right\}$, the utility vectors can be expressed as
	\[
		r_j[c, a] = |c|\sum_{\theta\in\Theta}p(\theta)\psi(c\vert\theta)u_i(a, \theta)\quad\forall c\in\C,\,\forall a\in\A.
	\]
	Note that each entry of such vectors can now be computed efficiently without iterating on all signal profiles belonging to each class, simply by exploiting Equation~\eqref{eq:classcard}. This gives the desired result.
	
	\paragraph{Computation of $\Amat_\ag^\varphi$.}
	To show that the coefficients of the matrix $\Amat_\ag^\varphi$ can be computed efficiently, we show that it is possible to determine in an efficient way the term $\sum_{\svec\in c} \mathbbm{1}[(\varphi(s_i), \svec_{-i})\in c^\prime]$ for any two $c,c^\prime\in\C$.
	
	Let $\cvec, \cvec^\prime\in\Naturals^{|S|}$ be the vectors describing classes $c$ and $c^\prime$, respectively. Trivially, if $||\cvec - \cvec^\prime||_1 > 2$, then, $\forall \varphi\in\Phi$, $\forall \svec\in c$, $(\varphi(s_i), \svec_{-i})\notin c^\prime$, since it is never the case that, by means of a single player deviation, a signal profile $\svec\in c$ is transformed into a signal profile $\svec^\prime\in c^\prime$. If $||\cvec - \cvec^\prime||_1 \leq 2$, then we can be in one of two cases, namely $||\cvec - \cvec^\prime||_1 = 2$ or $||\cvec - \cvec^\prime||_1 = 0$. Let us first consider the case $||\cvec - \cvec^\prime||_1 = 2$. Intuitively, $||\cvec - \cvec^\prime||_1 = 2$ implies that there exist two signals $s$ and $s^\prime$ such that $c[s] = c^\prime[s] + 1$ and $c[s^\prime] = c^\prime[s^\prime] - 1$. For notational convenience, let $\iota(\cvec, \cvec^\prime)=s$. Thus, we have $\sum_{\svec\in c} \mathbbm{1}[(\varphi(s_i), \svec_{-i})\in c^\prime]\neq 0$ if and only if $\varphi(s) = s^\prime$. In particular, when $\varphi(\iota(\cvec, \cvec^\prime)) = s^\prime$, the number of tuples $\svec\in c$ that can be transformed in tuples $\svec^\prime\in c^\prime$ corresponds to the number of tuples that have element $\iota(\cvec, \cvec^\prime)$ as the $i$-th element, which can be computed as 
	\[
	\sum_{\svec\in c} \mathbbm{1}[(\varphi(s_i), \svec_{-i})\in c^\prime] = \frac{\left(n-1\right)!}{(c[\iota(\cvec, \cvec^\prime)]-1)!\prod_{s^{\prime\prime}\neq \iota(\cvec, \cvec^\prime)}c[s^{\prime\prime}]!}.
	\]
	Finally, when $||\cvec - \cvec^\prime||_1 = 0$, \emph{i.e.,} when $\cvec = \cvec^\prime$, the number $\sum_{\svec\in c} \mathbbm{1}[(\varphi(s_i), \svec_{-i})\in c^\prime]$ is the number of tuples $\svec\in c$ that have as $i$-th element a signal $s_i$ such that $\varphi(s_i)=s_i$. Formally, this can be expressed as
	\[
		\sum_{\svec\in c} \mathbbm{1}[(\varphi(s_i), \svec_{-i})\in c^\prime] = \sum_{\substack{s\in S:\\\varphi(s)=s,\\ c[s] > 0}} \frac{\left(n-1\right)!}{(c[s]-1)!\prod_{s^{\prime\prime}\neq s}c[s^{\prime\prime}]!}.
	\]
	In conclusion, we have that
	\[
		\sum_{\svec\in c} \mathbbm{1}[(\varphi(s_i), \svec_{-i})\in c^\prime] =\begin{cases}
			0 & \text{if } ||\cvec - \cvec^\prime||_1 > 2 \\
			\frac{\left(n-1\right)!}{(c[\iota(\cvec, \cvec^\prime)]-1)!\prod_{s^{\prime\prime}\neq \iota(\cvec, \cvec^\prime)}c[s^{\prime\prime}]!} & \text{if } ||\cvec - \cvec^\prime||_1 = 2 \\
			\sum_{\substack{s\in S:\\\varphi(s)=s,\\ c[s] > 0}} \frac{\left(n-1\right)!}{(c[s]-1)!\prod_{s^{\prime\prime}\neq s}c[s^{\prime\prime}]!} &\text{otherwise.}
			
		\end{cases}
	\]
	This concludes the proof.
\end{proof}

\section{Proofs Omitted from Section \ref{sec:online}}
\thimp*
\begin{proof}
	Let us define two instances $\ia$ and $\ib$ of a game between a receiver $\med$ and a single sender $\ag$. The set of states of nature is composed by two states $\Theta=\left\{\theta_1, \theta_2\right\}$, the set of signals is $S=\left\{s_1, s_2\right\}$ and the set of actions is $\A=\left\{a, b\right\}$. Both instances share the same prior $p(\theta_1) = p(\theta_2) = 1/2$ and utility functions. In particular, given $\varepsilon > 0$, the utility functions are specified such that $u_\med(a, \theta_1) = u_\ag(a, \theta_1) = 1$, $u_\med(a, \theta_2) = u_\ag(a, \theta_2) = u_\med(b, \theta_1) =  u_\ag(b, \theta_1) = 0$, $u_\med(b, \theta_2) = 1$ and $u_\ag(b, \theta_2) = 2\varepsilon$. Furthermore, the two instances differ for the signaling schemes, which are defined as 
	\[
		\ia = \begin{cases}
			\psi_\ag(s_1\vert\theta_1) = 1 - 4\varepsilon \\
			\psi_\ag(s_2\vert\theta_1) =  4\varepsilon \\
			\psi_\ag(s_1\vert\theta_2) = 0 \\
			\psi_\ag(s_2\vert\theta_2) = 1,
		\end{cases}\quad
	\ib = \begin{cases}
		\psi_\ag(s_1\vert\theta_1) = 1 - 2\varepsilon \\
		\psi_\ag(s_2\vert\theta_1) =  2\varepsilon \\
		\psi_\ag(s_1\vert\theta_2) = 0 \\
		\psi_\ag(s_2\vert\theta_2) = 1.
	\end{cases}
	\]
	First, let us consider instance $\ia$. Fix a deviation function $\varphi\in\Phi$ such that $\varphi(s_1) = s_1$ and $\varphi(s_2) = s_1$. 
	By simple calculations it is possible to show that for each mechanism $\xivec\in\Xi$, the expected utility of the sender when she deviates according to $\varphi$ is the following 
	\[
		U_\ag^{\varphi}(\xivec) = \left(\frac{1}{2} - 2\varepsilon \right)\xi[s_1, a] + \varepsilon \xi^t[s_1, a] + \varepsilon.
	\]
	Similarly, the expected utility of the sender when she communicates truthfully the signal observed is 
	\[
		U_\ag(\xivec) = \left(\frac{1}{2} - 2\varepsilon \right)\xi[s_1, a]  - \varepsilon\xi[s_2, b] + 2\varepsilon.
	\]
	
	Given $T\in\Naturals_{>0}$ and $\delta\in (0,1)$, let us assume that, at each round $t\in [T]$ and with probability greater than $1-\delta$, the mechanism $\xivec^t$ selected by the the receiver in instance $\ia$ is IC. Hence, we can write the following 
	\[
		\mathbb{P}_\ia\left(\sum_{t\in [T]} \left[U_\ag^\varphi(\xivec^t) - U_\ag(\xivec^t)\right] \leq 0 \right) \geq 1-\delta,
	\] 
	which yields
	\[
		\mathbb{P}_\ia\left(\sum_{t\in [T]}\left[ \varepsilon\xi^t[s_1, a] + \varepsilon \xi^t[s_2, b] - \varepsilon\right] \leq 0 \right) \geq 1-\delta,
	\]
	where the subscript $\ia$ indicates the probability measure associated to instance $\ia$. The, we can leverage the Pinsker's inequality to state the following: 
		\begin{equation}\label{eq:imppinsk}
				\mathbb{P}_{\ib}\left(\sum_{t\in [T]}\left[ \varepsilon\xi^t[s_1, a] + \varepsilon \xi^t[s_2, b] - \varepsilon\right] \leq 0\right) \geq \mathbb{P}_\ia\left(\sum_{t\in [T]}\left[ \varepsilon\xi^t[s_1, a] + \varepsilon \xi^t[s_2, b] - \varepsilon\right] \leq 0 \right) -\sqrt{\frac{1}{2}\mc K(\ia, \ib)}, 
		\end{equation}
	where $\mc K\left(\ia, \ib\right)$ is the Kullback-Leibler divergence between instance $\ia$ and instance $\ib$.
	
	It is easy to check that the two instances differ only for the probability distributions of the receiver's rewards that arise when the signal sampled is $s_2$. In particular, when the signal sampled is $s_2$ and the action selected is $a$, the utility of the receiver follows a Bernoulli distribution $\mc B\left(\frac{2\varepsilon}{1/2 + 2\varepsilon}\right)$ in instance $\ia$, and a Bernoulli distribution $\mc B\left(\frac{\varepsilon}{1/2 + \varepsilon}\right)$ in instance $\ib$. Similarly, when the action selected is $b$, the receiver's utility follows a Bernoulli distribution $\mc B\left(\frac{1/2}{1/2 + 2\varepsilon}\right)$ in instance $\ia$, and a Bernoulli distribution $\mc B\left(\frac{1/2}{1/2 + \varepsilon}\right)$ in instance $\ib$. Then, exploiting the fact that $\mc K\left(\mc B(p), \mc B(q)\right) = \mc K(\mc B(1-p), \mc B(1-q))$, we can apply the KL decomposition Theorem \cite{lattimore2020bandit} to state the following
	\begin{align*}
			\mc K(\ia, \ib) &= \sum_{\substack{t\in [T]\\ s^t = s_2}}\left[\mathbb{E}_\ia\left[\xi^t[s_2, a]\right] + \mathbb{E}_\ia\left[\xi^t[s_2, b]\right] \right] \mc K\left(\mc B\left(\frac{2\varepsilon}{1/2 + 2\varepsilon} \right),\mc B\left(\frac{\varepsilon}{1/2 + \varepsilon}\right) \right) \\
			&\leq \sum_{\substack{t\in [T]\\ s^t = s_2}} 2\varepsilon \\
			&\leq 2\varepsilon T.
	\end{align*}
	Setting $\varepsilon = (0.01)^2/T$ and plugging into \Cref{eq:imppinsk}, we get
	\begin{equation}\label{eq:probbimp}
			\mathbb{P}_{\ib}\left(\sum_{t\in [T]}\left[ \varepsilon\xi^t[s_1, a] + \varepsilon \xi^t[s_2, b] - \varepsilon\right] \leq 0\right) \geq 0.99 - \delta.
	\end{equation}
	To conclude the proof we show that, with high probability, the receiver suffers linear regret in instance $\ib$. It is easy to check that the optimal IC mechanism in instance $\ib$ is the mechanism $\xivec^\star$ such that $\xi^\star[s_1, a] = 1$ and $\xi^\star[s_2, b] = 1$, yielding an expected utility to the receiver $U_\med(\xivec^\star) = 1-\varepsilon$. Hence, the cumulative regret can be expressed as 
	\begin{align*}
		R^T &= \sum_{t\in [T]}\left[1-\varepsilon - \left(\frac{1}{2} - \varepsilon\right)\left(\xi^t[s_1, a] + \xi^t[s_2, b]\right) - \varepsilon \right] \\
		&= \sum_{t\in [T]}\left[1 - 2\varepsilon - \left(\frac{1}{2} - \varepsilon\right)\left(\xi^t[s_1, a] + \xi^t[s_2, b]\right) \right].
	\end{align*}
	From \Cref{eq:probbimp} it follows that with probability grater than $0.99 - \delta$ 
	\[
		\sum_{t\in [T]}\left[\xi^t[s_1, a] + \xi^t[s_2, b] \right] \leq T,
	\]
	thus, with probability at least $0.99 - \delta$, it holds that
	\[
		R^T \geq \left(1 - 2\varepsilon \right) T - \frac{1}{2} \left(1 - 2\varepsilon \right) T = \frac{1}{2} \left(1 - 2\varepsilon \right) T,
	\]
	which gives the result.
\end{proof}
\section{Proofs Omitted from Section \ref{sec:full}}
\lemmaestfull*

\begin{proof}
	First, notice that for each $t \in [T]$, for all $c\in\C$ and for all $\theta\in\Theta$, it holds that
	\[
		\mathbb{E}\left[\mathbbm{1}\left[c(\svec^t) = c, \theta^t = \theta\right]\right] = \mathbb P\left(c, \theta\right) = p(\theta)\mathbb{P}\left(c\vert\theta\right) = \sum_{\svec\in c} p(\theta)\psi_{\ag}(\svec\vert\theta),
	\]
	where $\mathbb P\left(c, \theta\right)$ denotes the joint probability of sampling the state of nature $\theta$ and a signal profile belonging to class $c$, and $\mathbb{P}\left(c\vert\theta\right)$ is the probability of sampling a signal belonging to class $c$ given that the state of nature sampled is $\theta$. Then, for all $i\in\{\med, \ag\}$, for all $t\in [T]$, for all $c\in\C$ and for all $a\in\A$, it holds that
	\begin{align*}
		\mathbb{E}[\hat r_i^t[c, a]] &= \mathbb{E}\left[\frac{1}{t}\sum_{\theta\in\Theta} u_i(a,\theta)\sum_{\tau\in [t]} \mathbbm{1}\left[c(\svec^\tau) = c, \theta^\tau = \theta\right]\right] \\
		&= \sum_{\theta\in\Theta} u_i(a,\theta)\frac{1}{t}\sum_{\tau\in [t]} \mathbb{E}\left[\mathbbm{1}\left[c(\svec^\tau) = c, \theta^\tau = \theta\right]\right] \\
		&= \sum_{\theta\in\Theta}\sum_{\svec\in c} u_i(a,\theta) p(\theta)\psi_{\ag}(\svec\vert\theta) \\
		&= r_i[c, a].
	\end{align*}
	This concludes the proof.
\end{proof}

\lemmaconcfull*
\begin{proof}
	For each $\tau\in [T]$, for each $i\in\{\med,\ag\}$, for each $c\in\C$ and for each $a\in\A$, let $\tilde r_i^\tau[c, a] := \sum_{\theta\in\Theta} u_i(a,\theta) \mathbbm{1}\left[c(\svec^\tau) = c, \theta^\tau = \theta\right]$. Notice that $\mathbb{E}\left[\tilde \rvec_i^\tau\right] = \rvec_i$, and $\hat \rvec_i^t = \frac{1}{t}\sum_{\tau\in [t]} \tilde\rvec_i^\tau$, for each $t\in [T]$. Then, since it holds that $0\preceq \tilde\rvec_i^\tau\preceq 1$, where $\preceq$ denotes the element-wise $\leq$ operator, by Hoeffding's inequality we have that: for each $i\in\{\med,\ag\}$, for each $t\in [T]$, for each $c\in\C$ and for each $a\in\A$ it holds that 
	\[
		\mathbb{P}\left(| r_i[c, a] - \hat r_i^t[c,a]| \geq \sqrt{\frac{\log(4T |\C||\A|/\delta)}{2(t-1)}} \right) \leq \frac{\delta}{2T|\C||\A|}.
	\]
	The result follows by applying a union bound.
\end{proof}
\begin{lemma}\label{le:lemmautfull}
	Let $\rvec, \rvec^\prime\in\Reals^{|\C||\A|}$ be two vectors such that $||\rvec - \rvec^\prime||_{\infty}\leq \varepsilon$. Then, for each $\xivec\in\Xi$ it holds that: 
	\[
		| \xivec^\top \rvec - \xivec^\top\rvec^\prime | \leq |\C|\varepsilon.
	\]
\end{lemma}
\begin{proof}
	By Cauchy-Schwarz inequality it follows that: 
	\[
		|\xivec^\top (\rvec - \rvec^\prime)| \leq ||\xivec||_1 ||\rvec - \rvec^\prime||_{\infty}.
	\]
	The result follows since by assumption $||\rvec - \rvec^\prime||_{\infty}\leq \varepsilon$ and since 
	\[
		||\xivec||_1 = \sum_{c\in\C}\sum_{a\in\A}\xi[c, a] = \sum_{c\in\C} 1 = |\C|.
	\]
\end{proof}

\begin{lemma}\label{le:lemmadevfull}
	Let $\rvec, \rvec^\prime\in\Reals^{|\C||\A|}$ be two vectors such that $||\rvec - \rvec^\prime||_{\infty}\leq \varepsilon$. Then, for each $\varphi\in\Phi$ and $\xivec\in\Xi$ it holds that 
	\[
		\left\vert\xivec^\top \Amat_{\ag}^\varphi \rvec - \xivec^\top \Amat_{\ag}^\varphi \rvec^\prime \right\vert\leq |\C|\varepsilon.
	\]
\end{lemma}
\begin{proof}
	To prove this Lemma, we first show that for any $\xivec\in\Xi$ and $\varphi\in\Phi$, there exists a mechanism $\xivec^\prime\in\Xi$ such that $\xivec^\top \Amat_{\ag}^\varphi = \left(\xivec^\prime\right)^\top$ and then exploit Lemma \ref{le:lemmautfull} to obtain the desired result. 
	
	Notice that, by letting $\xivec^\prime$ be such that $\xivec^\top \Amat_{\ag}^\varphi = \left(\xivec^\prime\right)^\top$, we obtain that $\forall c\in\C$ and $a\in\A$,
	\[
	\xi^\prime[c, a] = \frac{1}{|c|}\sum_{\svec\in c}\xi[c(\psi(s_i), \svec_{-i}), a].
	\]
	Trivially, $\xivec^\prime\in\Reals^{|\C| |\A|}_{\geq 0}$. In order to show that $\xivec^\prime\in\Xi$, we need to show that $\forall c\in\C$, $\sum_{a\in\A}\xi^\prime[c, a] = 1$. To this extent, $\forall c\in\C$, we can write the following: 
	\begin{align*}
		\sum_{a\in\A} \xi^\prime[c, a] &= \frac{1}{|c|}\sum_{a\in\A}\sum_{\svec\in c} \xi[c(\psi(s_i), \svec_{-i}), a] \\
		&= \frac{1}{|c|}\sum_{\svec\in c} \sum_{a\in\A}\xi[c(\psi(s_i), \svec_{-i}), a] \\
		&= \frac{1}{|c|}\sum_{\svec\in c} 1 \\
		&= 1,
	\end{align*}
	which gives $\xivec^\prime\in\Xi$. Thus, 
	\begin{align*}
		\left\vert\xivec^\top \Amat_{\ag}^\varphi \rvec - \xivec^\top \Amat_{\ag}^\varphi \rvec^\prime \right\vert &= \left\vert\left(\xivec^\prime\right)^\top\rvec -  \left(\xivec^\prime\right)^\top\rvec^\prime \right\vert \\
		&\leq |\C|\varepsilon,
	\end{align*}
	where the last inequality follows from Lemma \ref{le:lemmautfull}. 
	
	This concludes the proof.
	\end{proof}

\begin{lemma}\label{le:lemmafeasiblefull}
	Let $\xivec^\star$ be an optimal solution to $\LP(0, \rvec_{\med}, \rvec_{\ag})$. Then, if event $\E_\delta^\full$ holds, $\xivec^\star$ is a feasible mechanism for $\LP(2|\C|\varepsilon_\delta^t, \hat\rvec_{\med}^t, \hat\rvec_{\ag}^t)$ for all $t\in [T]$.
\end{lemma}
\begin{proof}
	Given that event $\E^\full_\delta$ holds, we have that $\forall t\in[T]$, $||\rvec_{\ag} - \hat\rvec_{\ag}^{t}||_{\infty}\leq \varepsilon_\delta^t$.
	Thus, for all $t\in[T]$, by Lemma \ref{le:lemmautfull}, we have that 
	\[
	\left\vert(\xivec^\star)^\top \rvec_{\ag} - (\xivec^\star)^\top \hat\rvec_{\ag}^t \right\vert \leq |\C|\varepsilon_\delta^t.
	\]
	Furthermore, by Lemma \ref{le:lemmadevfull}, $\forall t\in[T]$ and $\forall\varphi\in\Phi$, 
	\[
	\left\vert\left(\xivec^\star \right)^\top\Amat_{\ag}^\varphi\rvec_{\ag}- \left(\xivec^\star \right)^\top\Amat_{\ag}^\varphi \hat\rvec_{\ag}^t\right\vert \leq |\C|\varepsilon_\delta^t.
	\]
	Hence, $\forall \varphi\in\Phi$, 
	\begin{align*}
		\left(\xivec^\star\right)^\top\left(\Amat_{\ag}^\varphi \hat\rvec_{\ag}^t - \hat\rvec_{\ag}^t\right) &\leq
		\left(\xivec^\star\right)^\top\left(\Amat_{\ag}^\varphi \rvec_{\ag} - \rvec_{\ag}\right) + 2 |\C|\varepsilon_\delta^t \\
		&\leq 2 |\C|\varepsilon_\delta^t,
	\end{align*}
	where the last inequality follows from the fact that, since $\xivec^\star$ is an optimal solution to $\LP(0, \rvec_{\med}, \rvec_{\ag})$, $\left(\xivec^\star\right)^\top\left(\Amat_{\ag}^\varphi \rvec_{\ag} - \rvec_{\ag}\right)\leq 0$ for all $\varphi\in\Phi$. This concludes the proof.
\end{proof}

\thfull*
\begin{proof}
	To prove this Theorem, we first derive the upper bound for the cumulative regret $R^T$ and then show the upper bound for cumulative IC violation $V^T$.
	Throughout this proof, we assume that event $\E_\delta^\full$ holds (let us remark that, by Lemma \ref{le:lemmaconcfull}, $\mathbb{P}\left(\E_\delta^\full\right)\geq 1-\delta$.
	\paragraph{Regret.}
	Recall the definition of cumulative regret:
	\[
	R^T = \sum_{t\in [T]} \left[\left(\xivec^\star\right)^\top \rvec_{\med} - \left(\xivec^t \right)^\top\rvec_{\med}\right].
	\]
	Since event $\E^\full_\delta$ holds, we have that $\forall t\in[T]$, $||\rvec_{\med} - \hat\rvec_{\med}^t ||_{\infty}\leq \varepsilon^t_\delta$. Then $\forall t\in [T]$, by Lemma \ref{le:lemmautfull}, we have that 
	\begin{align*}
		(\xivec^\star)^\top \rvec_{\med} &\leq (\xivec^\star)^\top \hat\rvec_{\med}^t + |\C|\varepsilon_{\delta}^t \\
		(\xivec^t)^\top \rvec_{\med} &\geq(\xivec^t)^\top \hat\rvec_{\med}^t - |\C|\varepsilon_{\delta}^t
	\end{align*} 
	Furthermore, since, by Lemma \ref{le:lemmafeasiblefull}, $\forall t\in [T]$ $\xivec^\star$ is a feasible mechanism for $\LP(2|\C|\varepsilon_\delta^t, \hat\rvec_{\med}^t, \hat\rvec_{\ag}^t)$, and since $\xivec^t$ is chosen as an optimal solution to the same linear optimization problem, we have that 
	\[
	\left(\xivec^\star\right)^\top \hat\rvec_{\med}^{t} \leq \left(\xivec^t\right)^\top \hat\rvec_{\med}^t.
	\]
	By plugging into the definition of cumulative regret we obtain the following: 
	\begin{align*}
		R^T &= \sum_{t\in [T]} \left[\left(\xivec^\star\right)^\top \rvec_{\med} - \left(\xivec^t \right)^\top\rvec_{\med}\right] \\
		&\leq \sum_{t\in [T]} \left[\left(\xivec^\star\right)^\top \hat\rvec_{\med}^t - \left(\xivec^t\right)^\top \hat\rvec_{\med}^t + 2|\C|\varepsilon_\delta^t\right] \\
		&\leq \sum_{t\in [T]} 2|\C|\varepsilon_\delta^t \\
		&\leq |\C|\sqrt{8T\log(4T|\C||\A|/\delta)},
	\end{align*}
	where the last inequality follows from $\sum_{t\in [T]}\frac{1}{\sqrt{t}} \leq 2\sqrt{T}$. 
	
	This concludes the first part of the proof.
	\paragraph{IC violation.}
	Recall the definition of cumulative IC violations 
	\[
	V^T = \sum_{t\in[T]} \left[\max_{\varphi\in\Phi} \left(\xivec^t\right)^\top\left(\Amat_{\ag}^\varphi\rvec_{\ag} 
	- \rvec_{\ag}\right)\right].
	\]
	Since event $\E^\full_\delta$ holds, we have that $\forall t\in[T]$, $||\rvec_{\med} - \hat\rvec_{\med}^t||_{\infty}\leq \varepsilon_\delta^t$. Then $\forall t\in [T]$, by Lemma \ref{le:lemmautfull}, we have that 
	\[
	\left(\xivec^t \right)^\top\rvec_{\ag} \geq \left(\xivec^t\right)^\top \hat\rvec_{\ag}^t - |\C|\varepsilon_\delta^t.
	\]
	Furthermore, by Lemma \ref{le:lemmadevfull}, $\forall \varphi\in\Phi$ 
	\[
	\left(\xivec^t\right)^\top\Amat_{\ag}^\varphi\rvec_{\ag} \leq \left(\xivec^t\right)^\top\Amat_{\ag}^\varphi\hat\rvec_{\ag}^t + |\C|\varepsilon_\delta^t.
	\]
	By plugging into the definition of $V^T$, we get
	\begin{align*}
		V^T &= \sum_{t\in[T]} \left[\max_{\varphi\in\Phi} \left(\xivec^t\right)^\top\left(\Amat_{\ag}^\varphi\rvec_{\ag} 
		- \rvec_{\ag}\right)\right]\\
		&\leq \sum_{t\in[T]} \left[\max_{\varphi\in\Phi} \left(\xivec^t\right)^\top\left(\Amat_{\ag}^\varphi\hat\rvec_{\ag}^t 
		- \hat\rvec_{\ag}^t\right) + 2 |\C|\varepsilon_\delta^t\right].
	\end{align*}
	Notice that, since  $\forall t\in [T]$ $\xivec^t$ is chosen as an optimal solution to $\LP(2|\C|\varepsilon_\delta^t, \hat\rvec_{\med}^t, \hat\rvec_{\ag}^t)$, 
	\[
	\max_{\varphi\in\Phi} \left(\xivec^t\right)^\top\left(\Amat_{\ag}^\varphi\hat\rvec_{\ag}^t 
	- \hat\rvec_{\ag}^t\right) \leq 2|\C|\varepsilon_\delta^t.
	\]
	Thus, 
	\begin{align*}
		V^T &\leq \sum_{t\in[T]} \left[\max_{\varphi\in\Phi} \left(\xivec^t\right)^\top\left(\Amat_{\ag}^\varphi\hat\rvec_{\ag}^t 
		- \hat\rvec_{\ag}^t\right) + 2 |\C|\varepsilon_\delta^t\right] \\
		&\leq \sum_{t\in [T]}4|\C|\varepsilon_\delta^t\\
		&\leq |\C|\sqrt{16T\log(4T|\C||\A|/\delta)},
	\end{align*}
	where the last inequality follows from $\sum_{t\in [T]}\frac{1}{\sqrt{t}} \leq 2\sqrt{T}$.
	
	This concludes the proof.
\end{proof}
\section{Proofs Omitted from Section \ref{sec:bandit}}

\lemmaestband*

\begin{proof}
	Fix a round $t\in [T]$ and $j\in\left\{\med, \ag\right\}$. Let us first notice that $\forall \tau\in [t]$ and for any $c\in\C$ that $\pi^t[c, a^\tau]= 1$
	\begin{align*}
		\mathbb{E}\left[u_j(a^\tau, \theta^\tau)\mathbbm{1}\left[s^\tau\in c\right]\right] &= \sum_{\svec\in c}\sum_{\theta\in\Theta}\mathbb{P}\left(\theta^\tau=\theta, \svec^\tau=\svec\right) u_j(a, \theta) \\
		&= \sum_{\svec\in c}\sum_{\theta\in\Theta}p(\theta)\psi_\ag(\svec\vert\theta)u_j(a, \theta) \\ 
		&= r_j[c, a].
	\end{align*}
	Furthermore, recall that $N^t[c, a]$ is defined such that $N^t[c, a] = \sum_{\tau\in [t]}\pi^\tau[c, a]$. Thus, for any class $c\in\C$ and action $a\in\A$, 
	\begin{align*}
		\mathbb{E}\left[\hat r^t[c, a]\right] &= \mathbb{E}\left[\frac{1}{N^t[c, a]}\sum_{\substack{\tau\in [t]\\\pi^\tau[c, a]=1}} u_j(a^\tau, \theta^\tau)\mathbbm{1}\left[s^\tau\in c\right]\right] \\
		&= r_j[c, a].
	\end{align*}
	This concludes the proof.
\end{proof}

\lemmaconcbandit*
\begin{proof}
	For $\delta\in (0,1)$, given that for any $j\in\left\{\med,\ag\right\}$ and $t\in [T]$, $\hat\rvec_j$ is an unbiased estimator of $\rvec^{\psi_\ag}_j$ (Lemma \ref{le:lemmaestbandit}) we can use Hoeffding's inequality to state the following: 
	\[
		\mathbb{P}\left(|\hat r^t_j[c,a] - r_j[c, a]| \geq \sqrt{\frac{\log( 8T|\C||\A|/\delta)}{2 N^t[c, a]}}\right) \leq \frac{\delta}{4T|\C||\A|}\quad\forall t\in [T], \forall c\in\C, \forall a\in\A,\forall j\in\left\{\med,\ag\right\}.
	\]
	Then, we can apply a union bound on $[T]$, $\C$, $\A$ and $\left\{\med, \ag\right\}$ and obtain
	\[
		\mathbb{P}\left(|\hat r_j^t[c, a] - r_j[c, a]| \leq \sqrt{\frac{\log( 8T|\C||\A|/\delta)}{2 N^t[c, a]}}\quad\forall t\in [T], \forall c\in\C, \forall a\in\A,\forall j\in\left\{\med,\ag\right\}\right) \geq 1-\frac{\delta}{2},
	\]
	which gives the result.
\end{proof}

\begin{lemma}\label{le:lemmafeasibleband}
	Let $\xivec^\star$ be an optimal solution to $\LP(0, \rvec_{\med}, \rvec_{\ag})$. Furthermore, define $\nu \defeq 2  |\C|\sqrt{\frac{\log( 8T|\C||\A|/\delta)}{2 E}}$. Then, if event $\E_\delta^\bandit$ holds, $\xivec^\star$ is a feasible mechanism for $\LP(\nu, \hat\rvec_{\med}^t + \etavec_\delta^t, \hat\rvec_{\ag}^t)$ for all $t\in [|\A|E+1, T]$.
\end{lemma}
\begin{proof}
	Fix a round $t\in [|\A|E+1, T]$. 
	First, let us notice that the initial exploration phase guarantees that all the counters $N^t[c,a]$ are lower bounded. Formally:
	\[
	N^t[c,a] \geq E\quad\forall c\in\C,\,\forall a\in\A.
	\]
	Hence, we have that 
	\begin{equation}
		||\etavec^t_\delta||_\infty = \max_{\substack{c\in\C\\ a\in\A}} \sqrt{\frac{\log( 8T|\C||\A|/\delta)}{2 N^t[c,a]}} \leq \sqrt{\frac{\log( 8T|\C||\A|/\delta)}{2 E}}. \label{eq:normboundband}
	\end{equation}	
	Furthermore, when event $\E_\delta^\bandit$ is verified, we have that 
	\begin{equation}\label{eq:evband}
		||\hat\rvec^t_\ag - \rvec_\ag ||_\infty \leq ||\etavec^t_\delta||_\infty \leq \sqrt{\frac{\log( 8T|\C||\A|/\delta)}{2 E}}.
	\end{equation}	
	%
	%where $\preceq$ is the element-wise $\leq$ operator.
	Then, by \cref{le:lemmautfull}, it holds that for each $\xivec\in\Xi$
	\begin{align}
		\left\vert\xivec^\top \left(\hat\rvec^t_\ag - \rvec_\ag\right)\right\vert \leq |\C|\sqrt{\frac{\log( 8T|\C||\A|/\delta)}{2 E}} \label{eq:boundnormband}
	\end{align}
Additionally, by \cref{le:lemmadevfull}, for all $\varphi\in\Phi$, it holds that 
\[
	|(\xivec^\star)^\top\Amat_\ag^\varphi \rvec_\ag - (\xivec^\star)^\top\Amat_\ag^\varphi \hat\rvec_\ag^t| \leq |\C|\sqrt{\frac{\log( 8T|\C||\A|/\delta)}{2 E}}.
\] 
Thus, we get that for each $\varphi\in\Phi$
\[
	\left(\xivec^\star\right)^\top \left(\Amat_\ag^\varphi \hat\rvec_\ag^t - \hat\rvec_\ag^t\right) \leq \left(\xivec^\star\right)^\top \left(\Amat_\ag^\varphi \rvec_\ag - \rvec_\ag\right) + 2  |\C|\sqrt{\frac{\log( 8T|\C||\A|/\delta)}{2 E}} \leq 2  |\C|\sqrt{\frac{\log( 8T|\C||\A|/\delta)}{2 E}},
\]
where the last equation follows from the fact that $\xivec^\star$ is a feasible solution for $\LP(0, \rvec_{\med}, \rvec_{\ag})$. This concludes the proof.
\end{proof}

\begin{lemma}\label{le:icexpl}
	For any $E\in [T/\A ]$, the following holds:
	\[
		\sum_{t=1}^{|\A|E}\max_{\varphi\in\Phi} \left(\xivec^t\right)^\top\left(\Amat_\ag^\varphi\rvec_\ag 
		- \rvec_\ag\right) = 0 
	\] 
\end{lemma}
\begin{proof}
	Note that, for each $\xivec^t$ in the exploration phase, by definition there exists an action $a^\prime\in\A$ such that $\xi[c, a^\prime]=1$, for any $c\in\C$.
	Then, 
	\[
		\left(\xivec^t\right)^\top \rvec_\ag = \sum_{c\in\C}\sum_{\svec\in c}\sum_{\theta\in\Theta} p(\theta)\psi_\ag(\svec\vert\theta)u_\ag(a^\prime,\theta)
		= \sum_{\theta\in\Theta} p(\theta)u_\ag(a^\prime, \theta).
	\]
	Furthermore, similarly to the proof of Lemma \ref{le:lemmadevfull}, for each $\varphi\in\Phi$, we can define a mechanism $\xivec^\prime$ such that $\left(\xivec^\prime\right)^\top\rvec_\ag = \left(\xivec^t\right)^\top\Amat_\ag^\varphi\rvec_\ag$, 
	where 
	\[
	\xi^\prime[c, a] = \frac{1}{|c|}\sum_{\svec\in c}\xi^t[c(\varphi(s_i), \svec_{-i}), a]\quad\forall c\in\C\,\forall a\in\A.
	\]
	Therefore, by simple calculations, we have that $\forall c\in\C$, $\xi^\prime[c, a^\prime] = 1$ and $\xi^\prime[c, a] = 0$, $\forall a\neq a^\prime$, which gives
	\[
		\left(\xivec^t\right)^\top\Amat_\ag^\varphi\rvec_\ag = \sum_{c\in\C}\sum_{\svec\in c}\sum_{\theta\in\Theta} p(\theta)\psi_\ag(\svec\vert\theta)u_\ag(a^\prime,\theta)
		= \sum_{\theta\in\Theta} p(\theta)u_\ag(a^\prime, \theta).
	\]
	Putting all together, we have 
	\[
		\sum_{t=1}^{|\A|E}\max_{\varphi\in\Phi} \left(\xivec^t\right)^\top\left(\Amat_\ag^\varphi\rvec_\ag 
		- \rvec_\ag\right) = \sum_{\theta\in\Theta} p(\theta)u_\ag(a^\prime, \theta) - p(\theta)u_\ag(a^\prime, \theta) = 0,
	\]
	which gives the result.
\end{proof}

\begin{lemma}\label{lemma:azh}
	For any $\delta\in (0,1)$ and $E\in [T]$ with probability at least $1-\frac{\delta}{2}$, the following holds:
	\[
		\sum_{t=|\A|E+1}^T \left(\xivec^t\right)^\top \etavec^t_\delta \leq |\C|\A| \sqrt{2T\log( 8T|\C||\A|/\delta)}  + |\C||\A|\sqrt{2T\log(2/\delta)}
	\]
\end{lemma}
\begin{proof}
	As a first step towards proving this Lemma, we leverage the Azuma-Hoeffding inequality \cite{cesa2006prediction} to show that
	\[
		\sum_{t=|\A|E+1}^T \left(\xivec^t\right)^\top \etavec^t_\delta \leq \sum_{t=|\A|E+1}^T \left[\left(\pivec^t\right)^\top \etavec^t_\delta\right] + |\C||\A|\sqrt{2T\log(2/\delta)}
	\]
	In particular, for all $t\in [|\A|E+1, T]$, we can define the random variable $X^t \defeq \sum_{\tau={|\A|E+1}}^t\left[ \left(\pivec^\tau\right)^\top\etavec^\tau_\delta - \left(\xivec^\tau\right)^\top\etavec^\tau_\delta\right]$. Note that, since deterministic mechanisms $\pivec^t$ are sampled according from $\xivec^t$, we have that, for every $t\in[|\A|E+1, T]$,
	\[
		\mathbb{E}\left[\pivec^t\vert\mc F^{t-1}\right] = \xivec^t.
	\]
	where $\mc F^{t-1}$ is the filtration generated up to time $t-1$ during the online interaction between the receiver's algorithm and the senders. Thus, for each $t\in [|\A|E+1, T]$, $X^t$ is a martingale. Hence, for $\delta\in (0,1)$, from Azuma-Hoeffding inequality we can conclude that, with probability at least $1-\frac{\delta}{2}$, it holds:
	\[
		\sum_{t=|\A|E+1}^T \left(\xivec^t\right)^\top \etavec^t_\delta \leq \sum_{t=|\A|E+1}^T\left[ \left(\pivec^t\right)^\top \etavec^t_\delta\right] + |\C||\A|\sqrt{2T\log(2/\delta)},
	\]
	which gives the desired result.
	
	Then, we conclude the proof by deriving an upper bound to $\sum_{t=|\A|E+1}^T \left(\pivec^t\right)^\top \etavec^t_\delta$. We can state the following chain of inequalities:
	\begin{align}
		\sum_{t=|\A|E+1}^T \left(\pivec^t\right)^\top \etavec^t_\delta &= \sum_{t=|\A|E+1}^T \sum_{\substack{c\in\C \\ a \in\A}} \pi^t[c, a] \sqrt{\frac{\log( 8T|\C||\A|/\delta)}{2 N^t[c, a]}}\notag \\
		&= \sum_{\substack{c\in\C \\ a \in\A}} \sum_{t=|\A|E+1}^T  \pi^t[c, a] \sqrt{\frac{\log( 8T|\C||\A|/\delta)}{2 N^t[c, a]}}\notag \\
		&= \sum_{\substack{c\in\C \\ a \in\A}} \sum_{\substack{t\in [|\A|E+1, T]:\\ \pi^t[c,a]=1}}  \sqrt{\frac{\log( 8T|\C||\A|/\delta)}{2 N^t[c, a]}}\notag \\
		&= \sum_{\substack{c\in\C \\ a \in\A}} \sum_{t=N^{|\A|E+1}[c, a]}^{N^T[c,a]} \sqrt{\frac{\log( 8T|\C||\A|/\delta)}{2 t}} \label{eq:az1}\\
		&\leq \sum_{\substack{c\in\C \\ a \in\A}} \sqrt{2\log( 8T|\C||\A|/\delta)N^T[c, a]} \label{eq:az2}\\
		&\leq |\C|\A| \sqrt{2T\log( 8T|\C||\A|/\delta)}\notag,
	\end{align}
	where \Cref{eq:az1} follows from the definition of $N^t[c,a]$ and \Cref{eq:az2} follows from the fact that ${\sum_{t=1}^T\frac{1}{\sqrt{t}}\leq 2\sqrt{T}}$. Thus, for any $\delta\in (0,1)$, we can conclude that, with probability at least $1-\frac{\delta}{2}$, 
	\begin{align*}
		\sum_{t=|\A|E+1}^T \left(\xivec^t\right)^\top \etavec^t_\delta &\leq \sum_{t=|\A|E+1}^T \left[\left(\pivec^t\right)^\top \etavec^t_\delta\right] + |\C||\A|\sqrt{2T\log(2/\delta)} \\
		&\leq |\C|\A| \sqrt{2T\log( 8T|\C||\A|/\delta)}  + |\C||\A|\sqrt{2T\log(2/\delta)}.
	\end{align*}
	This concludes the proof.
\end{proof}

\thbandit*
\begin{proof}
	To prove this Theorem, we first derive the upper bound for the cumulative regret $R^T$ and then show the upper bound for cumulative IC violation $V^T$.
	For $\delta\in (0,1)$, let $\overline \E_\delta$ be the event defined such that
	\[
		\overline\E_\delta\defeq  \left\{\sum_{t=|\A|E+1}^T \left(\xivec^t\right)^\top \etavec^t_\delta \leq |\C|\A| \sqrt{2T\log( 8T|\C||\A|/\delta)}  + |\C||\A|\sqrt{2T\log(2/\delta)}\right\}.
	\]
	By Lemma \ref{lemma:azh} we have that event $\overline\E_\delta$ is verified with probability at least $1-\frac{\delta}{2}$. Furthermore, from Lemma \ref{le:lemmaconcbandit}, we know that event $\E_\delta^\bandit$ holds with probability greater than $1-\frac{\delta}{2}$. 
	Throughout this proof, we assume that both events 
	$\overline\E_\delta$ and $\E_\delta^\bandit$ hold. Hence, by a simple union bound, it is possible to show that 
	\[
		\mathbb{P}\left(\overline\E_\delta, \E^\bandit_\delta\right) \geq 1-\delta.
	\]
	\paragraph{Regret.}
	Notice that, when event $\E_\delta^\bandit$ is verified, we can write the following:
	\[
		\rvec_\med \preceq \hat \rvec_\med^t + \etavec_\delta^t\preceq\rvec_\med + 2\etavec_\delta^t\quad \forall t\in [T],
	\]
	where $\preceq$ denotes the elemnt-wise $\leq$ relationship.
	
	Hence, given that $\Xi\subset\Reals^{|\C| |\A|}_{\geq 0}$, the following inequalities follow:
	\begin{align}
		R^T &= \sum_{t\in [T]} \left[\left(\xivec^\star\right)^\top \rvec_{\med} - \left(\xivec^t \right)^\top\rvec_{\med}\right] \notag \\
		&= \sum_{t=1}^{|\A|E} \left[\left(\xivec^\star\right)^\top \rvec_{\med} - \left(\xivec^t \right)^\top\rvec_{\med}\right] + \sum_{t=|\A|E+1}^T \left[\left(\xivec^\star\right)^\top \rvec_{\med} - \left(\xivec^t \right)^\top\rvec_{\med}\right] \notag \\
		&\leq |\A|E + \sum_{t=|\A|E+1}^T \left[\left(\xivec^\star\right)^\top \rvec_{\med} - \left(\xivec^t \right)^\top\rvec_{\med}\right] \notag \\
		&\leq |\A|E + \sum_{t=|\A|E+1}^T \left[\left(\xivec^\star\right)^\top \left(\hat \rvec_\med^t + \etavec_\delta^t\right) - \left(\xivec^t \right)^\top\rvec_{\med}\right] \label{eq:upperboundclean} \\
		&\leq |\A|E + \sum_{t=|\A|E+1}^T \left[\left(\xivec^t\right)^\top \left(\hat \rvec_\med^t + \etavec_\delta^t\right) - \left(\xivec^t \right)^\top\rvec_{\med}\right] \label{eq:regrbandfeasible}\\
		&\leq |\A|E + \sum_{t=|\A|E+1}^T \left[\left(\xivec^t\right)^\top \left(\rvec_\med + 2\etavec_\delta^t\right) - \left(\xivec^t \right)^\top\rvec_{\med}\right] \notag \\
		&= |\A|E + 2\sum_{t=|\A|E+1}^T \left(\xivec^t \right)^\top \etavec_\delta^t \notag\\
		&\leq |\A|E + 2|\C|\A| \left(\sqrt{2T\log( 8T|\C||\A|/\delta)}  + \sqrt{2T\log(2/\delta)}\right) \label{eq:regrbandaz} \\
		&\leq |\A|E + |\C|\A| \left(\sqrt{32T\log( 8T|\C||\A|/\delta)}\right), \notag
	\end{align}
	where \Cref{eq:upperboundclean} follows from the fact that event $\E_\delta^{\bandit}$ holds, \Cref{eq:regrbandfeasible} is a consequence of the fact that $\xivec^\star$ and $\xivec^t$ are, respectively a feasible (by Lemma \ref{le:lemmafeasibleband}) and an optimal (by definition) mechanism for $\LP(\nu, \hat\rvec_{\med}^t + \etavec_\delta^t, \hat\rvec_{\ag}^t)$, and \Cref{eq:regrbandaz} follows from the assumption that event $\overline\E_\delta$ holds.
	
	This concludes the first part of the proof.
	\paragraph{IC violation.}
	Let us notice that the initial exploration phase guarantees that at each round $t\in [|\A|E+1, T]$, all the counters $N^t[c,a]$ are lower bounded. Formally:
	\[
		N^t[c,a] \geq E\quad\forall t\in [|\A|E+1, T],\,\forall c\in\C,\,\forall a\in\A.
	\]
	Thus, $\etavec^t_{\delta} \preceq \etavec^{|\A|E}_\delta$, $\forall t\in [|\A|E+1, T]$. Additionally, since event $\E_\delta^\bandit$ holds, we have that 
	\[
		\hat\rvec^t_\ag - \etavec^{|\A|E}_\delta \preceq  \hat\rvec^t_\ag - \etavec^t_\delta \preceq \rvec_\ag \preceq \hat\rvec^t_\ag + \etavec^t_\delta \preceq \hat\rvec^t_\ag + \etavec^{|\A|E}_\delta,
	\]	
	where $\preceq$ denotes the element-wise $\leq$ relationship.
	Furthermore, notice that the design of the exploration phase guarantees that for each $c\in\C$ and $a\in\A$,
	\[
		\eta^{|\A|E}_\delta[c, a] =\sqrt{\frac{\log(8T|\C||\A|/\delta)}{2E}}. 
	\]
	In the following, with a slight abuse of notation, we denote $\eta^{|\A|E}_\delta = \sqrt{\frac{\log(8T|\C||\A|/\delta)}{2E}}$. Therefore, by \cref{le:lemmautfull} it holds that for all $t\in[|\A|E + 1, T]$
	\[
		|(\xivec^t)^\top \rvec_{\ag} - (\xivec^t)^\top \hat\rvec_{\ag}^t | \leq |\C|\eta^{|\A|E}_\delta, 
	\] and by \cref{le:lemmadevfull}, for all $t\in[|\A|E + 1, T]$ and $\varphi\in\Phi$, it holds that
	\[
		|(\xivec^t)^\top A_{\ag}^\varphi \rvec_{\ag} - (\xivec^t)^\top A_{\ag}^\varphi \hat\rvec^t_{\ag} | \leq |\C|\eta^{|\A|E}_\delta.
	\]
	This allows us to obtain the following chain of inequalities: 
	\begin{align}
		V^T &= \sum_{t\in[T]} \left[\max_{\varphi\in\Phi} \left(\xivec^t\right)^\top\left(\Amat_\ag^\varphi\rvec_\ag
		- \rvec_\ag\right) \right] \notag\\
		&= \sum_{t=1}^{|\A|E} \left[\max_{\varphi\in\Phi} \left(\xivec^t\right)^\top\left(\Amat_\ag^\varphi\rvec_\ag
		- \rvec_\ag\right)\right] + \sum_{t=|\A|E+1}^T \left[\max_{\varphi\in\Phi} \left(\xivec^t\right)^\top\left(\Amat_\ag^\varphi\rvec_\ag
		- \rvec_\ag\right)\right] \notag \\
		&= \sum_{t=|\A|E+1}^T \left[\max_{\varphi\in\Phi} \left(\xivec^t\right)^\top\left(\Amat_\ag^\varphi\rvec_\ag
		- \rvec_\ag\right)\right] \label{eq:eqnodevexpl}\\
		&\leq \sum_{t=|\A|E+1}^T \left[\max_{\varphi\in\Phi} \left(\xivec^t\right)^\top\left(\Amat_\ag^\varphi\hat\rvec_\ag^t
		- \hat\rvec_\ag^t\right) + 2 |\C|\eta_\delta^{|\A|E}\right]\notag\\
		&\leq \sum_{t=|\A|E+1}^T 4|\C|\eta_\delta^{|\A|E}\label{eq:eqinccompiterate} \\
		&= 4T|\C| \sqrt{\frac{\log(8T|\C||\A|/\delta)}{2E}}\notag,
	\end{align}
	
	where \cref{eq:eqnodevexpl} follows from \cref{le:icexpl}, \cref{eq:eqinccompiterate} follows from the fact that $\xivec^t$ is chosen as an optimal solution to $\LP(\nu, \hat\rvec_{\med}^t + \etavec_\delta^t, \hat\rvec_{\ag}^t)$ with $\nu = 2|\C|\eta^{|\A|E}_\delta$. 
	This concludes the proof.
\end{proof}

\lowerbound*
\begin{proof}
	Let us introduce two instances of a game between a receiver $\med$ and a single sender $\ag$. The set of states of nature is composed by three distinct states $\Theta= \left\{\theta_1, \theta_2, \theta_3\right\}$ and the set of signals is $ S = \left\{s_1, s_2, s_3\right\}$. Both instances share the same signaling scheme $\psi$ such that, $\psi(s_1\vert\theta) = \psi(s_2\vert\theta) = 1/2$ and $\psi(s_3\vert\theta)=0$, for $\theta\in\left\{\theta_1,\theta_2\right\}$, while $\psi(s_3\vert\theta_3)=1$ and $\psi(s_1\vert\theta_3) = \psi(s_2\vert\theta_3) = 0$. The set of actions available to the receiver is composed of two actions, namely $\A=\left\{a, b\right\}$. The utility function is specified such that 
	\[
		u_\med(a, \theta) = 
		\begin{cases}
			1 & \text{if } \theta\in\left\{\theta_1, \theta_2\right\} \\
			0 & \text{if } \theta=\theta_3,
		\end{cases} \quad
		u_\med(b, \theta) = 
		\begin{cases}
			0 & \text{if } \theta\in\left\{\theta_1, \theta_2\right\} \\
			1 & \text{if } \theta=\theta_3,
		\end{cases}
	\]
	for the receiver, while
	\[
	u_\ag(a, \theta) = 
	\begin{cases}
		1/2 & \text{if } \theta\in\left\{\theta_1, \theta_2\right\} \\
		0 & \text{if } \theta=\theta_3,
	\end{cases} \quad
	u_\ag(b, \theta) = 
	\begin{cases}
		0 & \text{if } \theta\in\left\{\theta_1, \theta_2\right\} \\
		1 & \text{if } \theta=\theta_3,
	\end{cases}
	\]
	for the sender.
	
	The only difference between the two instances lies in the prior. In particular, given $\varepsilon >0$,
	\[
		\ia = \begin{cases}
			p(\theta_1) = \frac{1}{3} - \varepsilon \\
			p(\theta_2) = \frac{1}{3} + \varepsilon \\
			p(\theta_3) = \frac{1}{3},
		\end{cases}\quad
		\ib = \begin{cases}
			p(\theta_1) = \frac{1}{3} + \varepsilon \\
			p(\theta_2) = \frac{1}{3} - \varepsilon \\
			p(\theta_3) = \frac{1}{3},
		\end{cases}
	\]
	
	It is easy to verify that in instance $\ia$ the optimal IC mechanism $\xivec^\star$ for $\med$ prescribes to select action $a$ when the signal reported are $s_1$ and $s_2$, while it selects action $b$ when the signal reported is $s_3$. Thus, straightforwardly, we can express the cumulative regret $R^T_{\ia}$ in instance  $\ia$ as 
	\begin{equation}
		R^T_{\ia} = \frac{1}{3}\sum_{t\in [T]} \left[\xi^t[s_1, b] + \xi^t[s_2, b] + \xi^t[s_3, a]\right].
	\end{equation}
	Let $\mathbb{P}_{i}$ be the probability measure associated to instance $i\in\left\{\ia,\ib\right\}$. 
	For $\delta\in (0,1)$, let us assume that, under probability measure $\mathbb P_{\ia}$, there exists an algorithm that guarantees $R_{\ia}^T \leq \frac{K}{3}$ with probability at least $1-\delta$, \emph{i.e.,} such that
	\[
		\mathbb{P}_{\ia}\left(\sum_{t\in [T]} \left[\xi^t[s_1, b] + \xi^t[s_2, b] + \xi^t[s_3, a]\right] \leq K\right) \geq 1 - \delta.
	\]
	Then, from the Pinsker's inequality we know that 
	\begin{align}
		\mathbb{P}_{\ib}&\left(\sum_{t\in [T]} \left[\xi^t[s_1, b] + \xi^t[s_2, b] + \xi^t[s_3, a]\right] \leq K\right) \notag \\ &\geq \mathbb{P}_{\ia}\left(\sum_{t\in [T]} \left[\xi^t[s_1, b] + \xi^t[s_2, b] + \xi^t[s_3, a]\right] \leq K\right) -\sqrt{\frac{1}{2}\mc K(\ia, \ib)} \notag \\
		&\geq 1- \delta - \sqrt{\frac{1}{2}\mc K(\ia, \ib)}\label{eq:low3},
	\end{align}
	where $\mc K(\ib, \ia)$ is the Kullback-Leibler divergence between instance $\ib$ and instance $\ia$. 
	
	Let us notice that the two instances are identical, except for what concerns the sender's reward distribution obtained when the signal state sampled is $s_1$ or $s_2$ and the action played is $b$. More in detail, we have that in instance $\ia$ such a reward distribution is a Bernoulli random variable with parameter $\frac{1}{2} + \frac{3}{2}\varepsilon$, while in instance $\ib$ it is a Bernoulli with parameter $\frac{1}{2} - \frac{3}{2}\varepsilon$. Hence, by applying the divergence decomposition theorem \cite{lattimore2020bandit} we obtain the following: 
	\begin{align}
		\mc K\left(\ia, \ib\right) &= \left(\sum_{\substack{t\in [T],\\ s^t = s_1}}\left[\mathbb{E}_{\ia}\left[\xi^t[s_1, b]\right]\right] + \sum_{\substack{t\in [T],\\ s^t = s_2}}\left[\mathbb{E}_{\ia}\left[\xi^t[s_2, b]\right]\right]\right)\mc K \left(\mc B\left(\frac{1}{2} + \frac{3}{2}\varepsilon\right), \mc B\left(\frac{1}{2} - \frac{3}{2}\varepsilon\right) \right) \notag\\
		&\leq
		\sum_{t\in [T]}\mathbb{E}_{\ia}\left[\xi^t[s_1, b] + \xi^t[s_2, b] + \xi^t[s_3, a]\right]\mc K \left(\mc B\left(\frac{1}{2} + \frac{3}{2}\varepsilon\right), \mc B\left(\frac{1}{2} - \frac{3}{2}\varepsilon\right) \right) \notag\\
		&\leq \mathbb{E}_{\ia}\left[ \sum_{t\in [T]}\left[\xi^t[s_1, b] + \xi^t[s_2, b] + \xi^t[s_3, a]\right]\right]\mc K \left(\mc B\left(\frac{1}{2} + \frac{3}{2}\varepsilon\right), \mc B\left(\frac{1}{2} - \frac{3}{2}\varepsilon\right) \right)\notag \\
		&\leq \mathbb{E}_{\ia}\left[ \sum_{t\in [T]}\left[\xi^t[s_1, b] + \xi^t[s_2, b] + \xi^t[s_3, a]\right]\right] 36\varepsilon^2\label{eq:low2}.
	\end{align}
	Additionally, from reverse Markov inequality it follows that 
	\begin{align}
		\mathbb{E}_{\ia}\left[ \sum_{t\in [T]}\left[\xi^t[s_1, b] + \xi^t[s_2, b] + \xi^t[s_3, a]\right]\right] &\leq \left(T -K\right)\mathbb{P}_\ia\left(\sum_{t\in [T]}\left[\xi^t[s_1, b] + \xi^t[s_2, b] + \xi^t[s_3, a]\right] \geq K\right) + K \notag \\
		&\leq \left(T-K\right)\delta + K. \label{eq:low1}
	\end{align}
	Thus, plugging inequalities \eqref{eq:low2} and \eqref{eq:low1} into \Cref{eq:low3}, we obtain
	\begin{equation}
		\mathbb{P}_{\ib}\left(\sum_{t\in [T]} \left[\xi^t[s_1, b] + \xi^t[s_2, b] + \xi^t[s_3, a]\right] \leq K\right) \geq 1-\delta - 3\varepsilon\sqrt{2\left(T-K\right)\delta + 2K}.	\label{eq:low4}
	\end{equation}

	Now let us focus on the formulation of the definition of the cumulative IC deviation $V^T_{\ib}$ for instance $\ib$. Let $\varphi\in\Phi$ be such that $\varphi(s_i) = s_3$, for any $s_i\in S$.   Trivially, by considering the IC deviation computed with respect to $\varphi$ at every round $t\in [T]$, we can provide a lower bound to $V_{\ib}^T$. Formally, 
	\[
		V_{\ib}^T \geq \sum_{t=1}^T U^\varphi_{\ib,\ag}(\xivec^t) - U_{\ib,\ag}(\xivec^t).
	\]
	By simple calculations it is possible to show that 
	\[
		U_{\ib,\ag}(\xivec^t) = \frac{2}{3} + \frac{\varepsilon}{2} \xi^t[s_1, b] + \frac{\varepsilon}{2} \xi^t[s_2, b] - \frac{1}{3} \xi^t[s_3, a],
	\]
	and 
	\[
		 U^\varphi_{\ib,\ag}(\xivec^t) = \frac{2}{3} + \varepsilon - \varepsilon \xi[s_3, a] - \frac{1}{3} \xi[s_3, a].
	\]
	Thus, we can write the following:
	\begin{align*}
		V_{\ib}^T &\geq \sum_{t=1}^T \left[U^\varphi_{\ib,\ag}(\xivec^t) - U_{\ib,\ag}(\xivec^t)\right] \\
		&= \varepsilon\sum_{t=1}^T \left[1 - \frac{1}{2}\xi^t[s_1, b] - \frac{1}{2}\xi^t[s_2, b] - \xi^t[s_3, a]\right] \\
		&\geq \varepsilon\sum_{t=1}^T \left[1 - \xi^t[s_1, b] - \xi^t[s_2, b] - \xi^t[s_3, a]\right] \\
		&= \varepsilon T - \varepsilon \sum_{t=1}^T \left[\xi^t[s_1, b] + \xi^t[s_2, b] + \xi^t[s_3, a]\right],
	\end{align*}
	which, together with \Cref{eq:low4}, gives us that, with probability at least $1-\delta - 3\varepsilon\sqrt{2\left(T-K\right)\delta + 2K}$, 
	\[
		V^T_\ib \geq \varepsilon T - \varepsilon K. 
	\] 
	To recover the desired result, set $\varepsilon = \frac{T^{-\alpha/2}}{100}$ and $K = \frac{T^\alpha}{9}$. Then, assuming without loss of generality $\delta\leq \frac{T^{\alpha - 1}}{9}$, we get 
	\[
		\mathbb{P}_\ib\left(V^T_\ib \geq 	\frac{T^{1-\alpha/2}}{100}\right) \geq 0.98 - \delta.
	\]
	This concludes the proof.
\end{proof}

\end{document}